\documentclass[12pt, a4paper]{article}
\usepackage[utf8]{inputenc}

\usepackage[utf8]{inputenc}
\usepackage[margin=1in]{geometry}
\usepackage{mathtools}
\usepackage{amsmath}
\usepackage{amsthm}
\usepackage{amssymb}
\usepackage{setspace}
\usepackage{booktabs}
\usepackage{tabularx}
\usepackage{amsmath} 
\usepackage{amssymb}
\usepackage{bm}
\usepackage{ascmac}
\usepackage{setspace}
\usepackage[at]{easylist}
\usepackage[normalem]{ulem}
\usepackage{comment}

\mathtoolsset{showonlyrefs}  

\makeatletter

\usepackage{amsthm}
\usepackage{amsfonts}
\usepackage{bbm}
\usepackage{graphics}
\usepackage{graphicx}
\usepackage{enumerate}
\usepackage{float}
\usepackage{caption}
\usepackage{subcaption}
\usepackage{natbib}

\theoremstyle{plain}
\newtheorem{theorem}{Theorem}

\theoremstyle{remark}

\usepackage{xspace}

\newcommand{\alg}[1]{\textbf{#1}\xspace}
\newcommand{\onesided}{\alg{One-sided}}
\newcommand{\DA}{\alg{DA}}
\newcommand{\ECDA}{\alg{ECDA}}

\usepackage[hyphens]{url}
\usepackage[colorlinks,urlcolor=blue, citecolor=blue, menucolor=blue]{hyperref}

\theoremstyle{plain}

\usepackage{algorithm,algpseudocode}

\usepackage{amsmath}

\allowdisplaybreaks

\makeatother
\title{Integrating Predictive Models into Two-Sided Recommendations: A Matching-Theoretic Approach\footnote{This work has been supported by JST ERATO Grant Number JPMJER2301, Japan, MiDATA Co., Ltd., and LINKBAL Co., Ltd. The study protocol was approved by the Institutional Review Board of the University of Tokyo (IRB No.~E2025ALS233) and pre-registered in the AEA RCT Registry (AEARCTR-0015446).}}
\author{
  Kazuki Sekiya\thanks{\href{mailto:kazuki-sekiya@g.ecc.u-tokyo.ac.jp}{kazuki-sekiya@g.ecc.u-tokyo.ac.jp}, Graduate School of Economics, The University of Tokyo} \and
  Suguru Otani\thanks{\href{mailto:suguru.otani@e.u-tokyo.ac.jp}{suguru.otani@e.u-tokyo.ac.jp}, Graduate School of Economics, The University of Tokyo} \and
  Yuki Komatsu\thanks{\href{mailto:komatsu-yuki286@g.ecc.u-tokyo.ac.jp}{komatsu-yuki286@g.ecc.u-tokyo.ac.jp}, Department of Economics, The University of Tokyo} \and
  Sachio Ohkawa\thanks{\href{mailto:sachio.okawa@midata.co.jp}{sachio.okawa@midata.co.jp}, MiDATA Co., Ltd.} \and
  Shunya Noda\thanks{\href{mailto:shunya.noda@e.u-tokyo.ac.jp}{shunya.noda@e.u-tokyo.ac.jp}, Graduate School of Economics, The University of Tokyo}
}

\date{\today}

\begin{document}

\maketitle

\begin{abstract}
    Two-sided platforms must recommend users to users, where matches (termed \emph{dates} in this paper) require mutual interest and activity on both sides. Naive ranking by predicted dating probabilities concentrates exposure on a small subset of highly responsive users, generating congestion and overstating efficiency. We model recommendation as a many-to-many matching problem and design integrators that map predicted login, like, and reciprocation probabilities into recommendations under attention constraints. We introduce \emph{effective dates}, a congestion-adjusted metric that discounts matches involving overloaded receivers. We then propose \emph{exposure-constrained deferred acceptance} (ECDA), which limits receiver exposure in terms of expected likes or dates rather than headcount. Using production-grade predictions from a large Japanese dating platform, we show in calibrated simulations that ECDA increases effective dates and receiver-side dating probability despite reducing total dates. A large-scale regional field experiment confirms these effects in practice, indicating that exposure control improves equity and early-stage matching efficiency without harming downstream engagement.
\end{abstract}

\vspace{0.4cm}


\section{Introduction}

\emph{Two-sided platforms} mediate matchings between distinct groups of users, such as men and women on online dating services \citep[e.g.,][]{hitsch2010matching} and marriage markets \citep{otani2025marriage}, employers and job seekers on job platforms \citep{horton2017effects}, or clients and freelancers in online marketplaces \citep{horton2010online,kassi2018online,kanayama2024nonparametric}. These platforms have become ubiquitous in the modern economy and often host numerous participants. Thus, enabling users to efficiently identify promising counterparts has become a central design challenge. Recommender systems play a crucial role in this process, as they shape users' search behavior and ultimately determine how effectively matches are formed.

In two-sided platforms, however, recommendation differs fundamentally from the setting studied in conventional recommender system literature. The objects being recommended are users, rather than passive items such as products or content. Predicting the preferences of the user receiving recommendations, whom we call the \emph{proposer}, and displaying the most preferred candidates among potential counterparts, whom we call \emph{receivers}, is generally not sufficient to generate successful matches. Match formation requires mutual interest: receivers must also find proposers attractive, and both sides must be active on the platform at the relevant time. Therefore, an effective two-sided recommendation requires integrating predictions of multiple dimensions of user behavior, including preferences, responsiveness, and activity, into a recommendation. The design of such an \emph{integrator} is a key component of two-sided recommender systems.

A growing literature points out that naively recommending based on predicted preferences can lead to severe congestion, where a small subset of receivers is repeatedly recommended to many proposers. Congestion can generate inequality across users, reduce overall match or employment rates, and deteriorate user experience on the platform. Addressing these concerns requires integrators that explicitly control exposure and limit excessive concentration of recommendations. Despite their practical importance, several fundamental questions remain underexplored. How should upper bounds on recommendations be imposed to mitigate congestion while preserving match efficiency? Which performance metrics are most informative for evaluating integrators in two-sided markets? How do different integrator designs perform in real-world platform environments? 

In this paper, we study a data-driven framework to achieve efficiency and equity in two-sided recommendation by integrating machine learning-based predictive models with insights from matching theory. We design integrators as mechanisms that take predictions of individual user behaviors as input and generate recommendations as matchings between users. We propose a novel performance metric to evaluate the quality of each expected match, taking into account the receiver-side congestion induced by recommendations. We then validate the performance of the proposed recommenders through a numerical simulation calibrated by a real-world dataset provided by \emph{CoupLink}, a Japanese online dating service with over 1.5 million cumulative users. We further implement the integrator identified as the best-performing method in our simulations on CoupLink and conduct a field experiment to evaluate its performance in the real world.

The process through which matches are formed on CoupLink is entirely standard and closely resembles that observed in many other online dating services and job platforms. A proposer (primarily male) can send a \emph{like} to another receiver (primarily female). If the receiver, upon being notified, reciprocates with a \emph{relike}, a match (termed a \emph{date}) is formed, enabling communication on the platform. The platform predicts the probability of a proposer/receiver liking a receiver/proposer (called the \emph{like rate}) and a user using the service (called the \emph{login rate}).

The core design principle of CoupLink's current integrator is as follows. For each proposer--receiver pair, the platform can compute a \emph{dating rate} by multiplying the proposer's login rate, the proposer's like rate, the receiver's login rate, and the receiver's relike rate. Based on these dating rates, the integrator recommends to each proposer a fixed number of receivers with the highest values. This approach is common in conventional (one-sided) recommender problems. However, it raises serious concerns in two-sided settings, as a small subset of receivers who are likely to be liked by many proposers, log in frequently, and readily return likes tend to be recommended repeatedly, thereby generating substantial congestion on the receiver side.

Even when the same predictive models are taken as given, there is substantial flexibility in how the predictions are integrated into a recommendation. For instance, instead of generating the rank-order list (ROL) over proposer--receiver pairs by predicted dating rates, one could sort them using a more myopic conversion metric, such as like rates. To prevent excessive concentration of recommendations on top receivers, the platform may also impose upper bounds on how often a given receiver can be recommended to proposers. Doing so introduces constraints on both sides of the market: the number of receivers that can be displayed to a single proposer and the number of proposers to whom a single receiver can be recommended. As a result, the recommendation problem naturally takes the form of a many-to-many matching problem. While it is natural to consider adopting the deferred acceptance algorithm (\DA) as a way to determine such matchings, the stability and proposer-side strategy-proofness it delivers is not, by itself, essential for determining recommend-match outcomes, and alternative mechanisms warrant consideration. In particular, although \DA typically imposes capacity constraints based on the number of proposers a receiver can match with, there is no clear justification for restricting attention to headcount-based capacity limits as the appropriate form of constraint.

We make four key findings to enhance two-sided recommendations. First, we demonstrate that using the dating rate to generate ROL input for \DA significantly outperforms the use of the like rate. Second, we introduce a new metric, the \emph{effective date}, which explicitly accounts for receiver-side congestion. Specifically, to compute the number of effective dates, we assume that, if a receiver forms dates with multiple proposers, the receiver selects one of them uniformly at random to develop a deeper relationship. We refer only to such dates as effective dates, and this definition discounts the value of dates with congested receivers. Using this metric, we show that conventional approaches that greedily recommend receivers with high dating rates (\onesided) are no longer effective because they do not take into account the cost of congestion. Third, while effective dates can be improved by adopting integrators that mitigate receiver-side congestion, such as deferred acceptance (\DA) with an appropriate parameter setting, we show that further improvements are achieved by a new method we term \emph{exposure-constrained deferred acceptance} (\ECDA), in which receiver capacity is defined based on the expected number of likes or dates. \ECDA is computationally fast and is practical for large-scale dating platforms such as CoupLink. Fourth, we implement on CoupLink the \ECDA with the configuration that exhibits the highest performance in our calibrated simulations and evaluate its effectiveness in practice. We adopt a difference-in-differences design, using Kanto, the most populous region in Japan, as the treatment area, and Kansai and Tokai, the next most populous regions, as control areas. The experimental results show that, consistent with the simulation findings, the introduction of \ECDA leads to a more equitable distribution of likes and dates and increases the number of effective dates compared with the existing recommender, which is based on \onesided, whereas \ECDA achieves fewer \emph{predicted} dates at the recommendation stage than \onesided.

Our difference-in-differences estimates reveal that expected outcomes at the recommendation stage move in the direction predicted by the simulation and are statistically significant. For realized outcomes observed within two weeks after recommendation, the treatment increases the volume of likes on both proposer and receiver sides, indicating improved upstream engagement. When all receivers are included, however, the effect on realized dating probability is small and statistically indistinguishable from zero, and post-engagement outcomes (messaging) move in the opposite direction. This aggregate pattern is driven by the extreme right tail of the receiver distribution: the top 0.1\% of receiver-day observations were disproportionately attracting recommendations and likes despite exhibiting very poor pass-through to post-engagement communication, consistent with severe congestion. By reallocating exposure away from this tail, \ECDA reduces the flow of low-conversion likes accruing to the receivers, which mechanically attenuates average effects on late-stage outcomes when they are pooled with everyone else. Importantly, once we exclude the top 0.1\% receiver-days, the results become coherent: the treatment shows statistically significant positive effects on average effective dates, dating probabilities, and average likes, while the coefficient on average dates is positive but not statistically significant for the remaining 99.9\% of users. Taken together, the evidence suggests that \ECDA improves early-stage matching efficiency for the vast majority of users by correcting congestion at the top, while receivers with exceptionally high relike rates lose attention that was unlikely to translate into substantive communication in the first place.

\subsection{Related Literature}

Recommender systems have been extensively studied in machine learning and data mining, predominantly in one-sided environments where platforms rank passive items for users. This literature provides both theoretical foundations for optimal ranking policies and empirical evidence on the impact of recommendations on user outcomes \citep[e.g.,][]{compiani2024online,horton2017effects}.

Work on two-sided (reciprocal) recommenders is more recent and comparatively smaller.\footnote{For example, \citet{palomares2021reciprocal}, the survey paper of the reciprocal recommender systems (RPS), states that the RPS ``\emph{arose in the last decade,}'' and ``\emph{to our knowledge there are no exhaustive literature analyses on state-of-art RPS developments and their growing range of applications ever since.}''} Early contributions explicitly incorporated two-sidedness through the scoring rule \citep{pizzato2010recon}. Relatedly, the economics and operations research literature studies two-sided platforms and matching markets, including design features that restrict search or information \citep{kanoria2021facilitating}, competitive platform strategy \citep{halaburda2018competing}, extensions incorporating bilateral preferences \citep{shi2023optimal}, and assortment/priority design in two-sided environments \citep{aouad2023online,shi2022optimal,ashlagi2022assortment}. Our focus is complementary: taking modern prediction systems as given, we study how to design practical integrators that map multi-dimensional behavioral predictions into recommendations.

Congestion is a central concern in two-sided recommendation as it causes inefficiency and inequality. Congestion effects are also present in one-sided settings (e.g., concentrated applications or attention), motivating designs that penalize excessive concentration or manage demand \citep{naya2021designing,manshadi2023redesigning,horton2024reducing}. In online dating contexts, several recent papers incorporate matching-theoretic structure to improve equity: \citet{chen2023reducing} combine regression-based prediction with a transferable-utility matching framework \citep{choo2006marries,galichon2022cupid}; \citet{tomita2022matching} evaluate matching-based recommenders; \citet{su2022optimizing} rank pairs by social-welfare criteria; and \citet{tomita2024fairreciprocalrecommendationmatching} pursue fairness via Nash social welfare guarantees. While the literature has analyzed the harmful effects of excessive congestion, existing studies have not proposed methods to quantify its impact on inefficiency. Importantly, our setting highlights that congestion is not merely a distributional concern at early stages: post-engagement data suggest that a small set of highly responsive users disproportionately attracts recommendations and likes, yet exhibits extremely low pass-through to substantive communication (e.g., message exchanges), implying that excessive concentration can directly drive late-stage inefficiency. This observation motivates us to introduce effective dates, a novel outcome-based metric that explicitly incorporates congestion into the evaluation of match efficiency. Building on this metric, we study the design of integrators that improve effective dates and validate their performance through a large-scale field experiment.

Three additional aspects distinguish our approach. First, we model the multi-stage actions required for a date---activity, liking, and responsiveness---and integrate predicted login, like, and relike probabilities into the construction of rank-order lists and matching-based integrators. While many mechanism-design-inspired approaches proxy preferences using like probabilities alone, incorporating activity and responsiveness is crucial in practice; \citet{rios2023improving} is a notable exception that combines login probabilities and history-dependent effects within an assortment-optimization framework. Second, we treat receiver-side capacity not as an exogenous headcount constraint but as a design choice tied to exposure (e.g., expected likes or dates), which connects to work on congestion management and capacity adjustment in matching markets \citep{arnosti2021managing,manshadi2022online}. Finally, because scalability is essential for production systems, algorithmic efficiency is a first-order consideration; parallel and accelerated methods for stable matching reflect this concern \citep{nakada2024parallel}. Our integrators leverage the structure of dating-rate sorting to admit fast greedy implementations, enabling large-scale simulation and deployment on a real dating platform.

\section{Model}\label{sec: model}

\subsection{Setup}

We consider a two-sided platform in which $I$ \emph{proposers} (e.g., males) can potentially match with $J$ \emph{receivers} (e.g., females) and develop a recommender system that determines which set of receivers should be recommended to each proposer. We denote the set of proposers and receivers by $[I] = \{1, \dots, I\}$ and $[J] = \{1, \dots, J\}$, respectively. We represent a recommendation by a \emph{recommendation matrix}, $M = (M_{ij})_{i \in [I], j \in [J]}$, where $M_{ij} \in [0, 1]$ is the probability that receiver $j$ is recommended to proposer $i$.\footnote{Using Birkhoff's algorithm \citep{birkhoff1946tres}, the probabilistic recommendation defined here can always be decomposed into a deterministic recommendation (i.e., $M_{ij} \in \{0, 1\}$ for all $(i, j)$) that satisfies the cognitive capacity constraint.} 

While one might think that recommending more receivers to a proposer weakly increases the possibility of a suitable match, this approach can lead to the proposer's cognitive overload. The overload not only diminishes user satisfaction but also raises the probability of the proposer discarding the recommendations altogether. To limit the cognitive burden on proposers, we set proposer $i$'s \emph{cognitive capacity} $c_i \in \mathbb{Z}_{++}$ on the number of receivers recommended to each proposer $i \in [I]$. The recommendation matrix must satisfy the constraint $\sum_{j \in [J]} M_{ij} \le c_i$ for all $i \in [I]$. This constraint is treated as a hard constraint that all candidate integrators must satisfy, rather than being a subject for integrator design. In other words, we assume that each proposer $i$ reviews recommendations for up to $c_i$ receivers, and ignores any additional recommended receivers beyond this limit.\footnote{This assumption aligns with the framework of the position-based model \citep[see][for example]{chuklin2022click} commonly employed in recommender system design. In this model, the probability of a user examining an item decreases with its position in the recommendation list. Specifically, we assume that the proposer examines the $c_i$-th position with probability one, while the $(c_i + 1)$-th position with probability zero.}

We refer to the act of a proposer and a receiver mutually sending likes and forming a match as a \emph{date}, to distinguish it from the ``match'' as a recommendation decision made by a matching mechanism like \DA. As illustrated in Figure~\ref{fg:chart}, for a date to occur between proposer $i$ and receiver $j$, the following steps must all be satisfied: (i) the platform recommends receiver $j$ to proposer $i$; (ii) proposer $i$ logs into the platform; (iii) proposer $i$ sends a \emph{like} to receiver $j$; (iv) receiver $j$ logs into the platform; (v) receiver $j$, upon observing proposer $i$'s like, sends a like back. To distinguish this second like sent by the receiver from the initial like sent by the proposer, we refer to it as a \emph{relike}. The probability of (i) occurring is $M_{ij}$, which is the platform's control variable. The \emph{login rate} of user $s \in [I]\cup[J]$, denoted by $\lambda_{s}$, is the probability that user $s$ logs into the platform. The \emph{like rate} from proposer $i \in [I]$ to receiver $j \in [J]$, denoted by $\alpha_{ij}$, represents the probability that proposer $i$ sends a like to receiver $j$, given that proposer $i$ logs in. Likewise, the \emph{relike rate} from receiver $j$ to proposer $i$, denoted by $\beta_{ij}$, represents the probability that receiver $j$ sends a relike to proposer $i$ given that proposer $i$ sends a like to receiver $j$ and receiver $j$ logs in. Accordingly, the probability that proposer $i$ and receiver $j$ have a date is $\delta_{ij} M_{ij}$, where $\delta_{ij} := \lambda_i \alpha_{ij} \lambda_{j} \beta_{ij}$ is called the \emph{dating rate}.

\begin{figure}
    \centering
    \includegraphics[width=\linewidth]{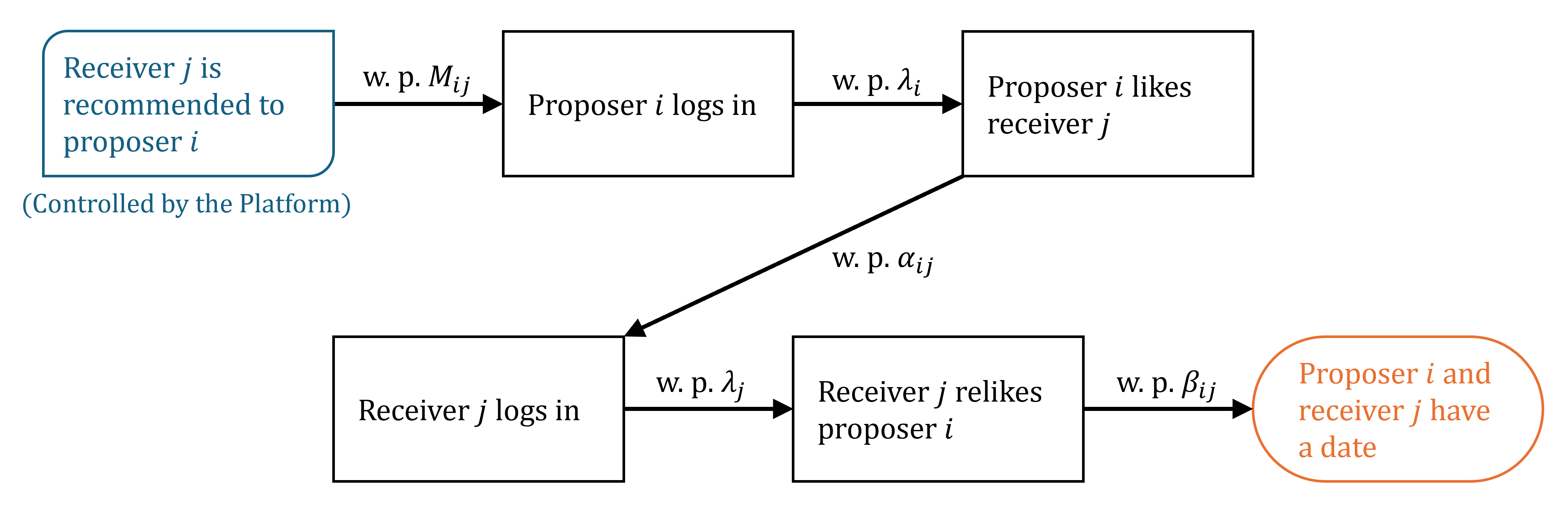}
    \caption{Chart Illustrating the Process from Recommendation to Date Formation}
    \label{fg:chart}
\end{figure}

In reality, these rates cannot be directly observed and must be predicted using data. These predictions are typically conducted using binary classification methods in machine learning, for which an extensive body of literature exists regarding methodologies and performance evaluation \citep[see, e.g.,][]{bishop2006pattern,hastie2009elements}. In this paper, to focus on the question of how these rates should be integrated into recommendations once provided, we proceed with the analysis under the assumption that the true (or sufficiently precise) values of these rates are known.

\subsection{Performance Indicators}\label{sec:performance_indicator}

Broadly speaking, the goal of the platform is to facilitate as many dates as possible among the largest number of users.
We primarily evaluate the quality of recommendations using the following three performance indicators.

\paragraph{Average Dates}
This metric is defined by
\begin{equation}
    \frac{1}{I}\sum_{i \in [I]} \sum_{j \in [J]} \delta_{ij} M_{ij} \:\: \text{(Proposer)} \quad \text{and} \quad \frac{1}{J}\sum_{i \in [I]} \sum_{j \in [J]} \delta_{ij} M_{ij} \:\: \text{(Receiver)}.
\end{equation}
Average dates evaluates the expected count of dates established on the platform. Since dates are desirable events, a higher number indicates better performance and serves as a fundamental welfare measure of efficiency. Note that average dates per proposer and receiver differ only by the relative sizes of the two sides and are therefore proportional to each other.

\paragraph{Average Effective Dates}
While average dates is a natural measure of efficiency, it does not account for congestion. In particular, unlike proposers, whose number of recommendation-matches is bounded by cognitive capacity, receivers can match with arbitrarily many proposers, and thus the receiver-side congestion can be substantial. In our setting, a ``date,'' mutual like that enables messaging, is only an intermediate outcome. When a receiver dates many proposers simultaneously, the likelihood that any given date develops into a deeper relationship is reduced.

To capture this effect, we introduce \emph{average effective dates}, a congestion-adjusted efficiency metric. We assume that, on each day, a receiver who forms multiple dates selects exactly one of them uniformly at random to develop into a deeper relationship. The uniform selection assumption is used solely to obtain a transparent congestion discount; we do not interpret it as a behavioral model of relationship formation. We refer to the selected date as an \emph{effective date}. Under this assumption, the expected contribution of a proposer--receiver pair $(i,j)$ to effective dates is discounted by the receiver's overall congestion.

Formally, let
\begin{equation}
    \mu_j \coloneqq \sum_{k \in [I]} \delta_{kj} M_{kj}
\end{equation}
denote the expected number of dates formed by receiver $j$. We define the \emph{effective dating rate} as
\begin{equation}
    \delta_{ij}^* \coloneqq \frac{1 - e^{-\mu_j}}{\mu_j}\delta_{ij}.
\end{equation}
The effective dating rate $\delta^*$ is computed by deriving the probability that a proposer and a receiver form an effective date, given a recommendation matrix $M$ and dating rates $\delta$. A formal derivation of $\delta_{ij}^*$ is provided in Appendix~\ref{sec:effective_dates_derivation}.
The factor $(1-e^{-\mu_j})/\mu_j$ is decreasing in $\mu_j$ and captures the reduction in match quality due to receiver-side congestion. 

We then evaluate a recommendation matrix $M$ by the \emph{average effective dates per proposer},
\begin{equation}
    \frac{1}{I} \sum_{i \in [I]} \sum_{j \in [J]} \delta_{ij}^* M_{ij}.
\end{equation}

\paragraph{Dating Probability}
CoupLink regards the \emph{dating probability}, i.e., the probability that a user has a date, as one of the most important performance indicators because whether a user has the opportunity to experience at least one date has a significant impact on user satisfaction and retention.

Formally, \emph{dating probability per proposer} and \emph{receiver} are defined as follows.
\begin{align}
    \text{(Proposer)} \hspace{2em}& \frac{1}{I}\sum_{i \in [I]} \lambda_i \left(1 - \prod_{j \in [J]} (1 - M_{ij} \lambda_j \alpha_{ij} \beta_{ij})\right),\\
    \text{(Receiver)} \hspace{2em}& \frac{1}{J}\sum_{j \in [J]} \lambda_j \left( 1 - \prod_{i \in [I]} (1 - M_{ij} \lambda_i \alpha_{ij} \beta_{ij})\right).
\end{align}
For each proposer $i$, $\lambda_i$ represents the probability that he logs in, while the expression inside the subsequent parentheses represents the probability that the proposer $i$ has at least one date, conditional on logging in. By multiplying these factors, we calculate the probability that proposer $i$ experiences a date as a result of the recommendation $M$. The intuition for the dating probability per receiver is similar.

Note that, as we have discussed, for a proposer, a date with a congested receiver has little value. In this sense, average effective dates and dating probability capture different dimensions: the former is a congestion-adjusted proposer-side valuation, while the latter is an extensive-margin probability of obtaining at least one date. Therefore, dating probability should mainly be used for evaluating receivers' welfare and equality.


\section{Rank-Order List}\label{sec:rank_order_list}
When determining recommendations based on the outputs of predictive models, we first define \emph{rank-order lists} (ROLs) that specify which receivers should be recommended to which proposers. These ROLs can be viewed either as the metrics used by the platform to construct recommendation rankings or, from a matching-theoretic perspective, as proxy representations of the preferences of proposers and receivers. 

In this paper, we consider two methods to construct ROLs.
\begin{description} 
    \item[Sorting by like rates] Each proposer $i$'s ROL ranks receivers by descending like rate $\alpha_{ij}$, while each receiver $j$'s ROL ranks proposers by descending relike rate $\beta_{ij}$.
    \item[Sorting by dating rates] Both proposers' and receivers' ROLs rank counterparts by descending $\delta_{ij} = \lambda_i \alpha_{ij} \lambda_j \beta_{ij}$.
\end{description}

\emph{Sorting by like rates} is the method most conceptually aligned with the design principles of matching theory. It determines the ROLs for proposing and accepting based purely on behavioral predictions related to users' like/relike actions. 

Even if the like rate is assumed to reflect a proposer's utility from a date, the utility should be discounted if the receiver may not be active or responsive. Accordingly, we further evaluate \emph{sorting by dating rates}, allowing for recommendations that better reflect the probability of forming a date. By considering both parties' preferences, this method is expected to yield recommendations that are more likely to result in successful dating.

\section{Integrators}\label{sec: integrator}

An \emph{integrator} is a function that maps the estimated values of the proposers' and receivers' login rates $(\lambda_s)_{s \in [I]\cup[J]}$, the proposer's like rate $(\alpha_{ij})_{i \in [I], j \in [J]}$, and the receiver's relike rate $(\beta_{ij})_{i \in [I], j \in [J]}$ (provided by predictive models), together with the proposer's cognitive capacity $c = (c_i)_{i \in [I]}$, to a recommendation $M$.

We consider the following integrators in our analysis: (i) A \emph{one-sided integrator} (\onesided), (ii) \emph{deferred acceptance} (\DA), and (iii) \emph{exposure-constrained DA} (\ECDA).

\subsection{One-Sided Integrator}

\onesided is designed purely to assist proposers in identifying the receivers they prefer. This integrator recommends receivers to each proposer $i$ in descending order, according to proposer $i$'s ROL, up to the proposer's cognitive capacity $c_i$. In particular, although it incorporates several minor adjustments, CoupLink's current mechanism is close to \onesided that generates ROLs based on dating rates. The platform takes no intermediary role in the matching process, except for sorting. While naive, it is a conventional method for generating recommendations and a useful baseline for analysis.

On two-sided platforms, a key weakness of \onesided is their tendency to concentrate recommendations on a small subset of receivers. On these platforms, there typically exist receivers who attract likes or dates with high probability regardless of which proposer they are recommended to. When ROLs are constructed from such ratings, recommendations naturally concentrate on these receivers. Under sorting by like rates, whether a receiver is recommended is independent of the receiver's preferences over proposers, and heavily recommended receivers may therefore find notifications from proposers in whom the receiver has little interest burdensome. Sorting by dating rates may partially mitigate this issue; however, even if a receiver forms dates with many proposers through mutual likes, constraints on time and attention limit meaningful communication, making effective dates less likely. This reduction in effective dates also harms proposers. Moreover, congestion also generates receivers who are not recommended sufficiently many times to obtain dating opportunities. For these reasons, even if \onesided optimizes short-run conversions such as likes or dates, it is not well-suited for achieving higher-level objectives.

\subsection{Deferred Acceptance (\DA)}

A direct and fundamental way to avoid congestion on a small subset of receivers is to impose a capacity (upper bound) on the number of proposers to whom each receiver can be recommended at a given time. By introducing such a constraint, exposure that would otherwise be concentrated on a few receivers is reallocated to others, leading to a more equitable distribution of recommendations and potentially generating more effective dates. While the number of receivers recommended to each proposer $i$ is already limited by cognitive capacity $c_i$, imposing a capacity $q_j$ on how many times each receiver $j$ can be recommended to proposers transforms recommendation generation into a many-to-many matching problem.

Once capacity is imposed on the receiver side, it becomes necessary to ration access to preferred receivers among proposers whose demand exceeds that capacity. Specifically, the platform must determine which proposers are allowed to be recommended as highly ranked receivers in their ROLs. A natural way to implement such rationing is to use receivers' ROLs: each receiver tentatively keeps proposers up to her capacity according to her ROL and rejects the rest, after which rejected proposers apply to other receivers until their own capacities are filled. This procedure is identical to deferred acceptance (\DA), the standard algorithm for producing stable matchings in two-sided matching problems. In the context of two-sided recommendation, \DA can therefore be interpreted as an integrator that determines which receivers are recommended to which proposers, given receiver-side capacity constraints.

In general, computing \DA may be insufficiently fast when the platform is very large, making it potentially unsuitable for recommendation system design. However, when ROLs are constructed by sorting on dating rates, the same outcome can be computed much more efficiently using an alternative algorithm. This is because the utility from each proposer--receiver pair $(i, j)$ is given by $\delta_{ij}$ and thus aligned for both the proposer and the receiver. This property ensures that the behavior of a simple greedy algorithm coincides exactly with the behavior of \DA.
\begin{theorem}\label{thm: DA is equivalent to greedy}
    \DA with sorting by dating rates returns the same recommendation matrix $M$ as the following greedy algorithm: Initialize $M$ as an $I \times J$ matrix of zeros. Take $(i, j) \in [I] \times [J]$ in descending order of the dating rates $\delta_{ij}$. Match $(i, j)$ if both proposer $i$'s and receiver $j$'s capacities still allow for it, and skip otherwise.
\end{theorem}
Proofs are provided in Appendix~\ref{sec: proofs}. 

\subsection{Exposure-Constrained Deferred Acceptance (\ECDA)}

\DA aims to achieve equity by limiting the number of proposers to whom a receiver can be recommended. While defining a receiver's capacity in this manner is natural in the original formulation of \DA as a matching mechanism, this approach may not be optimal in our context. In a two-sided recommendation, the mechanism generates a recommendation matrix that determines ``which receivers are recommended to which proposers,'' rather than directly deciding ``who dates whom.'' Importantly, individual receivers can observe how many likes they receive, but cannot observe how many proposers they were recommended to. As a result, the number of proposers with whom a receiver is recommend-matched does not directly affect the user's experience. Instead, what the recommender system directly influences---and what receivers care about---is the number of likes they receive and the number of dates that are ultimately formed.

For this reason, each receiver $j$'s capacity should not be defined by the number of proposers they match with but rather by the expected total number of likes or dates they are predicted to receive, i.e., $\sum_{i \in [I]} M_{ij} \lambda_i \alpha_{ij} \le q_j$, or $\sum_{i \in [I]} M_{ij} \delta_{ij} \le q_j$, where $q_j$ is the receiver's capacity parameter. Because this variant of \DA defines receiver capacity using variables that are directly related to exposure---such as the number of likes received from proposers or the number of dates the receiver is expected to enjoy---we refer to it as \emph{Exposure-Constrained Deferred Acceptance} (\ECDA).

The difference between \ECDA and \DA lies only in how the receiver's capacity is defined. Since proposers' login rates and like rates are typically fractional, \ECDA requires extending the \DA protocol to allow fractional proposals and acceptances. Formally, \ECDA takes as input all users' ROLs, receivers' capacity $(q_j)_{j \in [J]}$, and \emph{exposure weight} $(w_{ij})_{i \in [I], j \in [J]}$. When we use the conditional probability that proposer $i$ sends a like to receiver $j$, i.e., take $w_{ij} = \lambda_i \alpha_{ij}$ as the exposure weight, we call it an \emph{\ECDA with like exposure}. When we use the conditional probability that $i$ and $j$ have a date, i.e., take $w_{ij} = \delta_{ij}$ as the exposure weight, we call it an \emph{\ECDA with date exposure}. Note that the conventional \DA corresponds to the case of $w_{ij} \equiv 1$ for all $i$ and $j$. \ECDA generates the recommendation matrix $M$ by iterating through the following steps.

\begin{enumerate}
    \item Select a proposer who has not yet tentatively matched with receivers up to his full cognitive capacity. This proposer identifies the highest-ranked receiver on his ROL who has not yet rejected him (even partially).
    \item The proposer proposes to the receiver with an amount equal to the smaller of (i) the remaining capacity for this proposer--receiver pair, which is given by one unit minus the amount they have already been tentatively matched, or (ii) the proposer's remaining capacity for additional matches.
    \item The receiver, say $j$, evaluates those currently tentatively matched and any new proposers in the order of her ROL. Each tentatively accepted proposer $i$ consumes $w_{ij}$ units of the receiver's capacity, and receiver $j$ tentatively accepts proposers in this order until her capacity $q_j$ is fully utilized. If the receiver's capacity is not sufficient to fully accept the next proposer on her ROL, then she accepts him fractionally. Any remaining proposers are rejected, becoming tentatively unmatched.
    \item Repeat steps (1)--(3) until one of the following termination conditions is met: (i) all proposers have fully utilized their cognitive capacity, or (ii) all receivers have been proposed to by every eligible proposer.
\end{enumerate}

Parallel to \DA, when ROLs are generated based on dating rates, \ECDA's behavior aligns with that of a greedy algorithm.
\begin{theorem}\label{thm: ECDA is equivalent to greedy}
    \ECDA with sorting by dating rates returns the same recommendation matrix $M$ as the following greedy algorithm: Initialize $M$ as an $I \times J$ matrix with all zeros. Take each pair $(i, j)$ in descending order of the dating rates $\delta_{ij}$. Update $M_{ij}$ as follows:
    \begin{equation}
        M_{ij} \leftarrow \min\left\{1, c_i - \sum_{k \in [J]}M_{ik}, \frac{q_j - \sum_{l \in [I]}w_{lj}M_{lj}}{w_{ij}}\right\}.
    \end{equation}
\end{theorem}

Without Theorems~\ref{thm: DA is equivalent to greedy} and \ref{thm: ECDA is equivalent to greedy}, \DA and \ECDA may not be implementable at the platform scale because they require the construction of rankings of all proposers and receivers, iterated proposals, and repeated updates of the set of proposers tentatively accepted by receivers. In contrast, when ROLs are based on dating rates, the \DA and \ECDA outcomes can be computed by a simple greedy pass over pairs sorted by $\delta_{ij}$, with only constant-time capacity updates per pair. The resulting implementation requires only a single sorting step (and a linear scan thereafter).

\section{Simulation}\label{sec:simulation}

\subsection{Predictive Models}
In this section, we conduct a simulation analysis based on real-world data from CoupLink to evaluate the practical performance of the integrator. We use the predictive models for login behavior as well as like and relike probabilities that CoupLink currently deploys in its production environment. These models are used in day-to-day business operations to generate recommendations on the platform and therefore represent product-quality prediction systems rather than experimental or simplified benchmarks. 

Parallel to the model presented in Section~\ref{sec: model}, CoupLink facilitates matches by allowing proposers to send ``likes'' to receivers of the opposite sex. If the receiver returns a ``relike'' upon being notified, a ``date'' is established, enabling matched users to exchange messages on the platform.\footnote{The terms ``relike'' and ``date'' are defined specifically for this paper and do not correspond to the official terminology used on CoupLink.} Although both men and women can act as proposers by sending likes, in practice, men predominantly take on this role. Throughout simulations, we fix men as the proposer side and women as the receiver side.

While users can manually search for candidates using a search engine, this study focuses on ``AI Recommendations,'' the platform's recommendation feature.\footnote{Note that the search feature available for manual searches also involves recommendations based on preference predictions.} When this feature is activated, for a user (who becomes a proposer) a fixed number of candidates (receivers) are displayed one at a time. Proposers can either send a like or skip to the next receiver. The recommendation list updates daily, and the recommender system determines each proposer's recommendation list.

CoupLink employs machine-learning models to predict key user behaviors that determine whether a recommendation results in a date. Specifically, the platform predicts (i) proposer login probability, (ii) proposer like probability, (iii) receiver login probability, and (iv) receiver relike probability, and integrates these predictions to compute dating rates used for recommendation generation.

The proposer's login rate is defined as the predicted probability that an active user (one who has logged in within the past one week and satisfies platform-specific eligibility criteria) will use the AI recommendation feature on the following day. Conditional on using the recommendation feature, the proposer's like rate is the predicted probability that the proposer chooses to send a like rather than skip a recommended receiver. On the receiver side, the login rate is defined as the predicted probability that an active user logs in within the next seven days, and the relike rate is the predicted probability that the receiver reciprocates a like upon observing it.

While the detailed feature set and model specifications are proprietary, all four predictive tasks are implemented as binary classification problems using gradient-boosted decision trees. These models are trained on user registration attributes and historical service usage logs, and are updated daily to reflect recent activity.

\begin{figure}[t]
\centering
\begin{minipage}[b]{0.47\columnwidth}
    \centering
    \includegraphics[width=1.0\columnwidth]{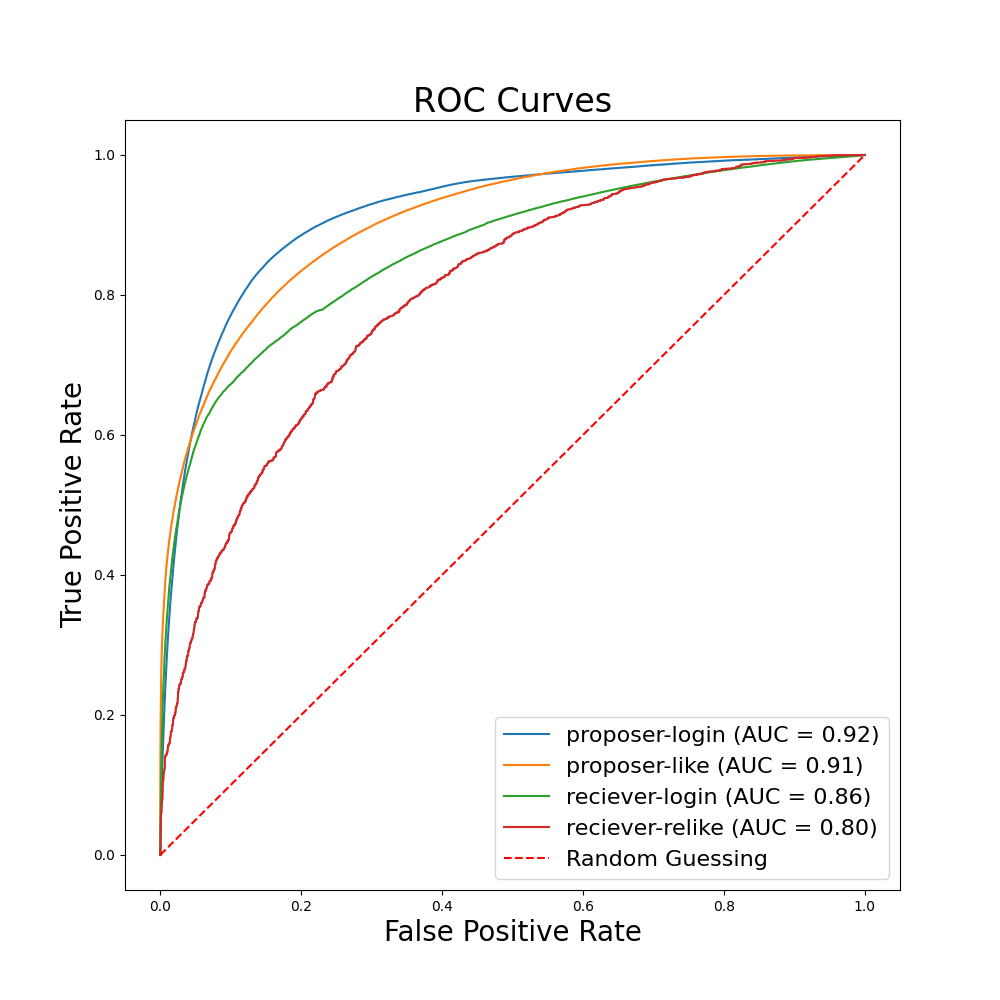}
    \subcaption{ROC Curves}
\end{minipage}
\hfill
\begin{minipage}[b]{0.47\columnwidth}
    \centering
    \includegraphics[width=1.0\columnwidth]{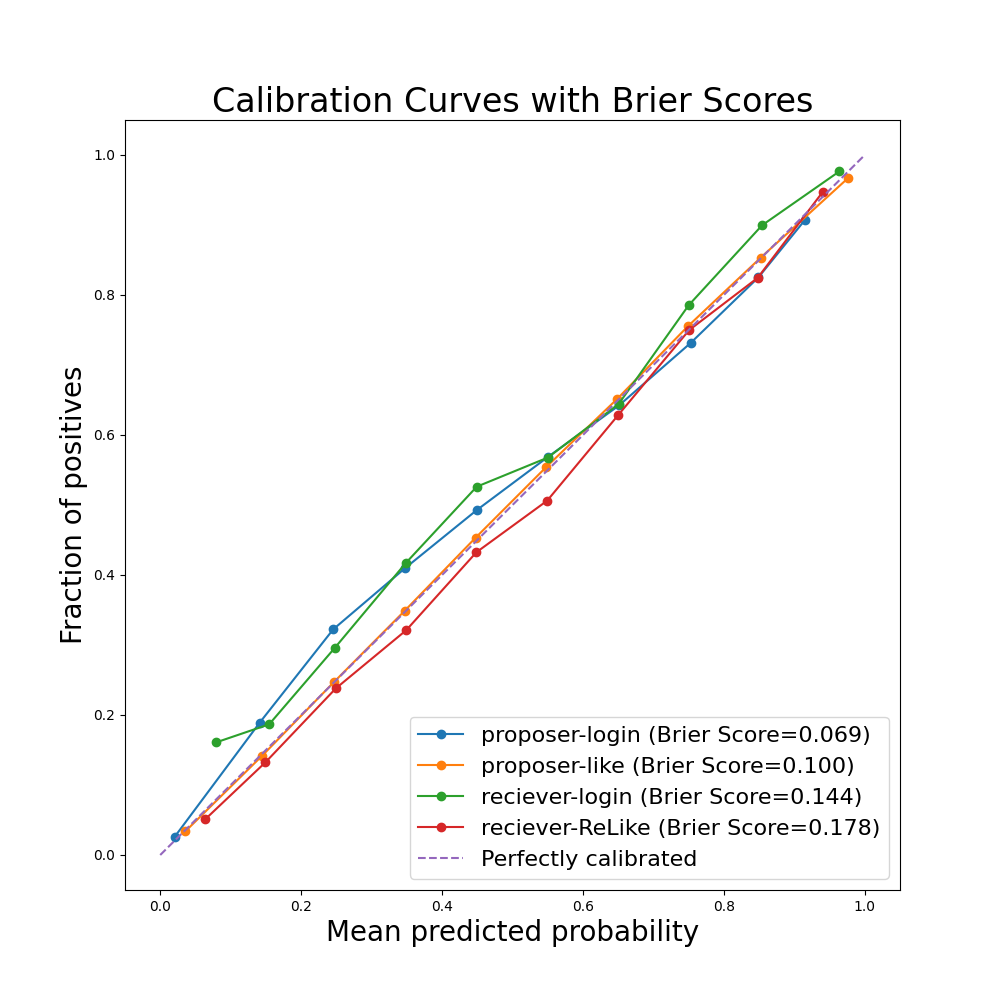}
    \subcaption{Calibration Curves}
\end{minipage}
\caption{ROC Curves and Calibration Curves}
\label{fg:roc_and_calibration_appendix}
\begin{flushleft}\footnotesize
  Note: The ROC curve plots the true positive rate against the false positive rate at various threshold settings. The calibration curve plots the mean predicted probability against observed positive rates across bins.
  \end{flushleft}
\end{figure}

\begin{table}[t]
    \centering
    \begin{tabular}{lcc}
        \hline
        Model       & AUC & Brier Score \\
        \hline
        Proposer-login          & 0.916        & 0.069                \\
        Proposer-like           & 0.906        & 0.100                \\
        Receiver-login          & 0.863        & 0.144                \\
        Receiver-relike         & 0.800        & 0.178                \\
        \hline
    \end{tabular}
    \caption{Area Under the Curve (AUC) and Brier Score}
    \label{tb:auc_brier_appendix}
    \begin{flushleft}\footnotesize
        Note: The Area Under the Curve (AUC) is the total area under the ROC curve, providing a single scalar value that reflects the classifier's performance. The Brier Score measures the mean squared difference between the predicted probabilities and the realized outcomes.
    \end{flushleft}
\end{table}

To assess whether these predictive models are sufficiently accurate for use in the simulation study, we evaluate their performance using standard classification and calibration metrics. Figure~\ref{fg:roc_and_calibration_appendix} reports receiver operating characteristic (ROC) curves and calibration curves for each model, while Table~\ref{tb:auc_brier_appendix} summarizes the area under the ROC curve (AUC) and Brier score.

The ROC curve illustrates a model's ability to distinguish positive and negative outcomes across different thresholds. An AUC value of 0.5 corresponds to random classification, while higher values indicate stronger discriminatory power. Across the four predictive tasks, the models achieve AUC values of 0.916 for proposer login, 0.906 for proposer like, 0.863 for receiver login, and 0.800 for receiver relike, indicating strong predictive performance.

Calibration curves assess how closely predicted probabilities align with observed frequencies. The Brier score provides a scalar summary of calibration accuracy, with lower values indicating better performance. The reported Brier scores---0.069 (proposer login), 0.100 (proposer like), 0.144 (receiver login), and 0.178 (receiver relike)---indicate reasonable calibration across tasks.

Taken together, these results confirm that the predictive models capture proposer and receiver behaviors with sufficient accuracy to serve as inputs for the simulation analysis. Accordingly, in the main analysis, we treat the predicted login, like, and relike probabilities as primitives and focus on how these predictions should be integrated into recommendation mechanisms.

\subsection{Simulation Environment}\label{subsec:simulation_environments}
We evaluate the proposed mechanisms using the following two simulation environments.

\paragraph{Synthetic Market}

To test the performance of sorting by like rates, which is not computationally efficient except for \onesided, we conduct simulations with a small, \emph{synthetic market}. We construct a joint distribution of predicted behavioral metrics---proposer login rate, proposer like rate, receiver login rate, and receiver relike rate---using the values predicted by CoupLink on November~6,~2024, a regular business day with no special event. Each rate is discretized using bins of width~0.1 to avoid an excessive number of empty cells and to keep the joint distribution at a manageable size.

From this joint distribution, we randomly draw $I=1,000$ proposers and $J=1,000$ receivers to form a single dataset, imposing symmetry in market size across the two sides. This normalization ensures that average outcomes are invariant to whether they are computed per proposer or per receiver; in practice, the two sides are also approximately balanced in the real data. The proposer's cognitive capacity is fixed at $c_i = 25$ receivers for all $i \in [I]$. We repeat this procedure ten times, resulting in ten independent datasets.

\paragraph{Empirical Market}

The \emph{empirical market} more closely simulates the business environment of CoupLink and thus is used for tuning parameters in the field experiment.
The dataset is constructed from proposer--receiver candidate pairs generated by CoupLink's production system on April~7,~2025. These pairs satisfy standard eligibility conditions used by the platform, such as geographic proximity and age differences, and are enriched with predicted proposer login rates, receiver login rates, like rates, relike rates, proposer area information, and the recommendation outcomes produced by the existing algorithm. For the simulations, we restrict attention to candidate pairs in which the proposer is located in the Kanto region. Independent of the platform's overall scale, we construct an experimental environment using a sampled set of users; the resulting dataset includes approximately $I = 8,000$ proposers and $J = 5,000$ receivers. Proposers' cognitive capacity is set to $c_i = 65$.

Note that, due to the scale, \DA and \ECDA that rely on sorting by like rates are not computationally practical for this environment, since the acceleration techniques developed in Theorems~\ref{thm: DA is equivalent to greedy} and \ref{thm: ECDA is equivalent to greedy} are not applicable.

\subsection{Results}\label{subsec:simulation_results}

\subsubsection{Synthetic Market}

\begin{table}[tb]
\centering
\caption{Performance Comparison of \onesided and \DA in the Synthetic Market}
\label{tab:sim_one_sided_da_ecda}
\begin{tabular}{lcccccc}
\hline
& Avg. & Avg. & \multicolumn{2}{c}{Dating Prob.} & Avg.\\
\cline{4-5}
Integrator
& Dates
& Effective Dates
& (Proposer)
& (Receiver)
& Likes \\
\hline
\onesided (like-sort)
& 0.2184
& 0.1859
& 0.1260
& 0.1716
& 4.182 \\
\DA (like-sort)
& 0.2227
& 0.1926
& 0.1269
& 0.1816
& 3.978 \\
\onesided (date-sort)
& 0.3543
& 0.2317
& 0.1527
& 0.2237
& 3.531 \\
\DA (date-sort)
& 0.3030
& 0.2312
& 0.1407
& 0.2233
& 3.736 \\
\DA (date-sort, cap=40)
& 0.3372
& 0.2436
& 0.1495
& 0.2340
& 3.675 \\
\hline
\end{tabular}

\begin{flushleft}\footnotesize
Notes: The capacities of \DA (like-sort) and \DA (date-sort) are set to $q_j = 25$ (balanced market).
\end{flushleft}
\end{table}

We begin by examining the performance of \DA as the most basic alternative to \onesided, which directly optimizes the expected number of likes or dates. The receiver capacity in \DA is set to $q_j = 25$, the same as the proposers' cognitive capacity. Table~\ref{tab:sim_one_sided_da_ecda} reports the resulting performance metrics for \onesided and \DA under two different constructions of ROLs: sorting by like rates and sorting by dating rates. To illustrate the effect of increasing the receiver's capacity, we also list the performance of \DA with sorting by dating rates and capacity $q_j = 40$.

Whereas sorting by like rates increases likes, sorting by dating rates clearly outperforms in dates, effective dates, and dating probabilities. Under \onesided, sorting by like rates leads to approximately 4.2 likes sent per proposer per day, which are also received by receivers on average. When sorting by dating rates is adopted instead, this number decreases to about 3.5 likes per day.
In contrast, the average dates per day increases substantially, from 0.22 under sorting by like rates to 0.35 under sorting by dating rates. A similar pattern is observed for \DA and is also confirmed for \ECDA across a wide range of parameters, although these results are omitted from the table. This difference arises because sorting by dating rates incorporates not only whether a proposer is likely to prefer a receiver, but also whether the receiver is likely to log in and whether she is likely to reciprocate the proposer's interest. Since dates represent a higher-level objective than likes, these results suggest that, as a baseline principle, sorting by dating rates yields superior performance in two-sided recommendation settings.

From a theoretical perspective, \DA is expected to improve effective dates relative to \onesided by mitigating congestion, even at the cost of generating fewer dates overall. However, the resulting levels of average effective dates under \DA are similar to those under \onesided. A closer look at the distribution of likes and dates under \DA reveals that, while \DA indeed alleviates congestion, a pronounced reduction in dates offsets this effect. Compared to \onesided, \DA produces significantly fewer dates, and the loss in date volume outweighs the gains from reduced congestion. Although \DA appears appealing for matching theorists, a naive implementation does not deliver superior performance.

Nevertheless, when receiver capacity $q_j$ is appropriately chosen, \DA performs well. The bottom row reports the performance of \DA\ with a relaxed capacity of up to 40. While average dates slightly decrease relative to \onesided, \DA outperforms \onesided in both average effective dates and dating probability per receiver. This highlights that for \DA to perform well in two-sided recommendation, it is crucial to appropriately tune receiver capacity, which is a design parameter rather than a primitive of the problem. Across all simulations conducted in this paper, we consistently find that overall performance is higher when receiver capacity is set to be relatively loose.

\begin{figure}[tb]
  \begin{center}
       \begin{minipage}[b]{0.47\columnwidth}
            \centering
            \includegraphics[width=1.0\columnwidth]{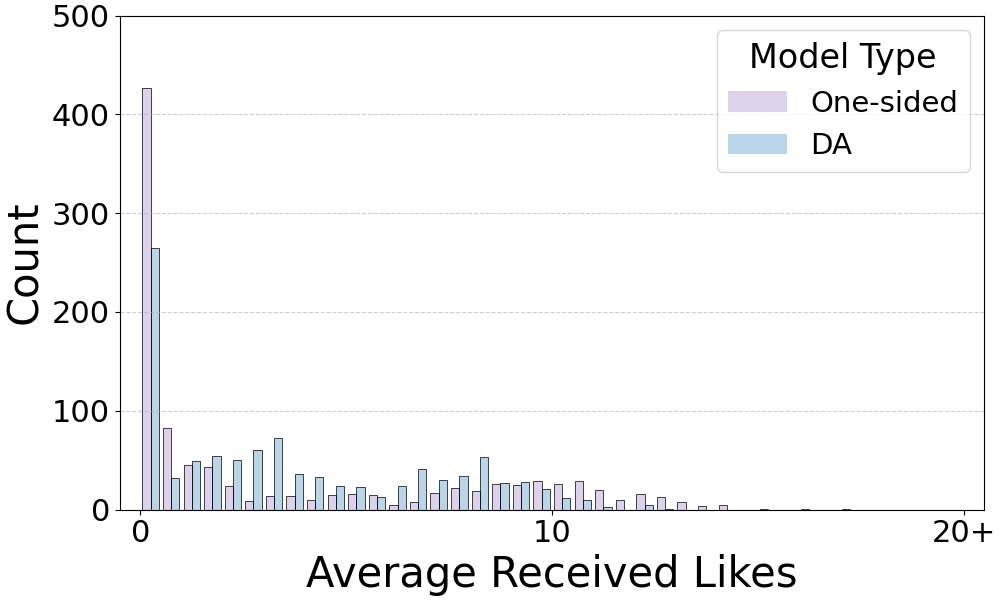}
            \subcaption{Likes, \onesided and \DA}
        \end{minipage}
        \hfill
        \begin{minipage}[b]{0.47\columnwidth}
            \centering
            \includegraphics[width=1.0\columnwidth]{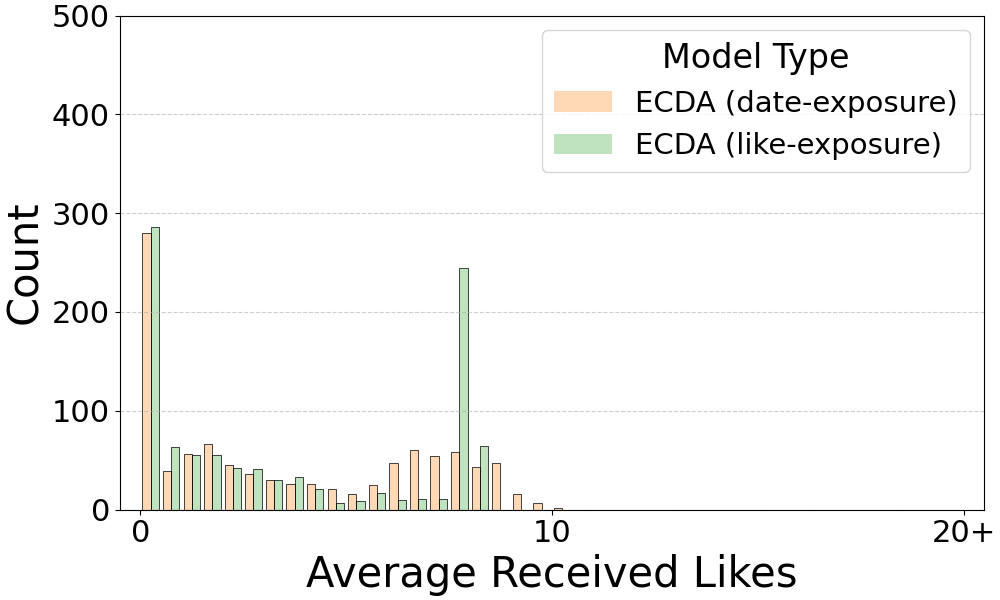}
            \subcaption{Likes, \ECDA}
        \end{minipage}\\
       \begin{minipage}[b]{0.47\columnwidth}
            \centering
            \includegraphics[width=1.0\columnwidth]{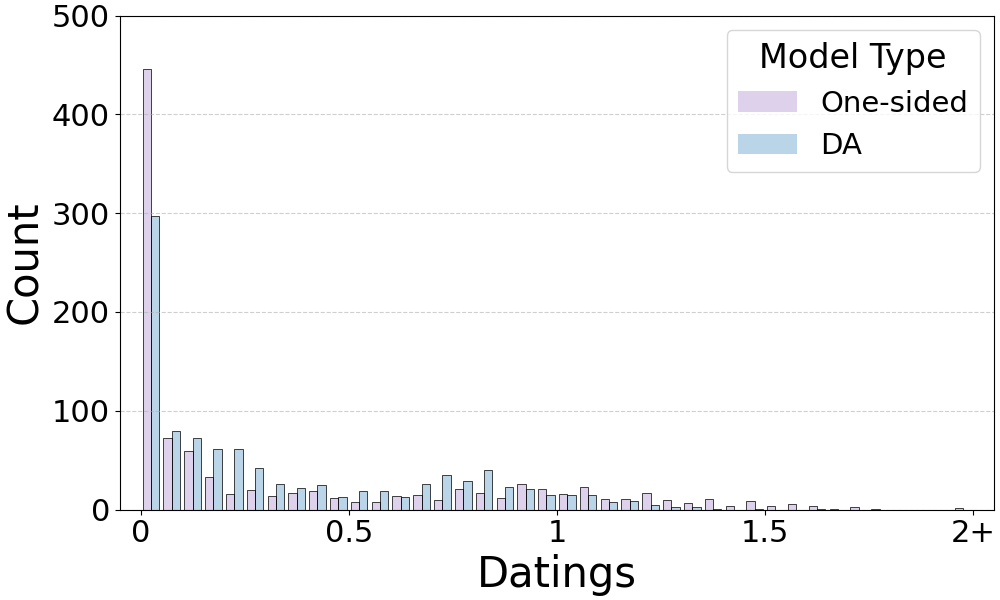}
            \subcaption{Dates, \onesided and \DA}
        \end{minipage}
        \hfill
        \begin{minipage}[b]{0.47\columnwidth}
            \centering
            \includegraphics[width=1.0\columnwidth]{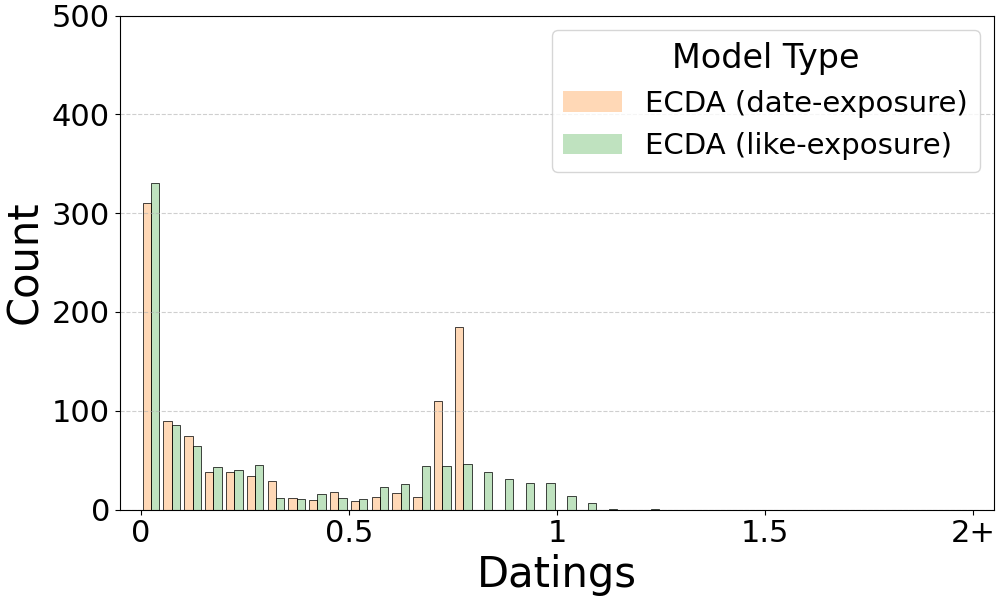}
            \subcaption{Dates, \ECDA}
        \end{minipage}
       \caption{The Distribution of Likes and Dates per Receiver (\onesided vs \DA vs \ECDA)}
       \label{fig:compare_quad}
  \end{center}
  \begin{flushleft}\footnotesize
  Note: ROLs are based on dating rates. $q_j = 40$ for \DA, $q_j = 8.0$ for \ECDA with like exposure, and $q_j = 0.75$ for \ECDA with date exposure.
  \end{flushleft}
\end{figure}

To examine how different integrators affect congestion, Figure~\ref{fig:compare_quad} illustrates the distributions of the number of likes received and the number of dates formed by individual receivers under \onesided, \DA, and \ECDA. ROLs are based on dating rates. For \DA and \ECDA, we adopt the capacity parameters that yielded the highest effective dates among those we tested: $q_j = 40$ for \DA, $q_j = 8.0$ for \ECDA with like exposure, and $q_j = 0.75$ for \ECDA with date exposure. Note that all of these receiver capacity settings are relatively loose. Under \DA, the average number of times a receiver is recommended is 25, while under \ECDA with like exposure the average number of likes received is 3.60, and under \ECDA with date exposure the average number of dates is 0.32.

Under \onesided, the distributions of both likes and dates exhibit the heaviest right tails. A large number of receivers receive no recommendations at all, resulting in zero expected likes and dates, while a non-negligible fraction of receivers accumulate an exceptionally large number of likes and dates. These receivers experience severe congestion, reflecting the tendency of \onesided to concentrate exposure on a small subset of receivers.

\DA partially alleviates this concentration by distributing recommendation opportunities more evenly across receivers. However, because receiver capacity in \DA is defined in terms of the number of proposers with whom a receiver can be recommend-matched, substantial dispersion in the realized numbers of likes and dates remains. As a result, congestion is mitigated but not eliminated.

By contrast, \ECDA directly constrains the predicted number of likes or dates received by each receiver, which allows it to compress the distributions of both outcomes into a substantially narrower range. Under like exposure, a visible spike appears at the point where the cap on expected likes binds, while under date exposure, a similar spike emerges at the cap on expected dates. Although the exposure constraint is imposed on only one of the two dimensions, likes and dates are closely related outcomes. Consequently, even the unconstrained metric exhibits considerably less dispersion under \ECDA than under \DA.

This reduction in dispersion implies that \ECDA succeeds in allocating likes and dates more equitably across receivers, avoiding extreme congestion without eliminating valuable recommendation opportunities. These distributional changes provide a direct explanation for \ECDA's superior performance in average effective dates and dating probability on the receiver side.

\subsubsection{Empirical Market}

\begin{table}[tb]
\centering
\caption{Performance Comparison of \onesided, \DA, \ECDA, and the Current Recommender}
\label{tab:actual_one_sided_da_ecda}
\begin{tabular}{lcccccc}
\hline
& Avg. & Avg. & \multicolumn{2}{c}{Dating Prob.} & Avg.\\
\cline{4-5}
Integrator
& Dates
& Effective Dates
& (Proposer)
& (Receiver)
& Likes \\
\hline
\onesided (like-sort)
& 0.0401
& 0.0241
& 0.0288
& 0.0344
& 5.384 \\
\onesided (date-sort)
& 0.1231
& 0.0579
& 0.0508
& 0.0857
& 3.693 \\
Current
& 0.1182
& 0.0584
& 0.0493
& 0.0863
& 3.565 \\
\DA (date-sort, cap=140)
& 0.1033
& 0.0605
& 0.0424
& 0.0904
& 3.610 \\
\ECDA (like-exposure, cap=22.0)
& 0.1054
& 0.0603
& 0.0424
& 0.0912
& 3.308 \\
\ECDA (date-exposure, cap=1.5)
& 0.0928
& 0.0623
& 0.0401
& 0.0932
& 3.658 \\
\hline
\end{tabular}

\begin{flushleft}\footnotesize
Notes: \ECDA (like-exposure) and \ECDA (date-exposure) adopt sorting by dating rates.
\end{flushleft}
\end{table}

We first summarize the performance of \onesided, \DA, \ECDA, and CoupLink's current recommender in Table~\ref{tab:actual_one_sided_da_ecda}. For \DA and \ECDA, we fix ROLs to sorting by dating rates. The receiver capacity parameters are the optimal ones among those we have tested: $q_j = 140$ for \DA, $q_j = 22.0$ for \ECDA with like exposure, and $q_j = 1.5$ for \ECDA with date exposure. Across all designs, exposure constraints trade off total dates for effective dates, and \ECDA with sorting by dating rates and date exposure achieves the best effective dates and dating probability per receiver.

Excluding average likes, sorting by dating rates substantially outperforms sorting by like rates across all date-related outcomes. While sorting by like rates yields a higher average number of likes (5.384) than sorting by dating rates (3.693), it performs far worse in terms of average dates (0.0401 versus 0.1231). Given this large gap in the number of realized dates, similarly pronounced differences emerge for effective dates and dating probabilities. This observation is consistent with that from the synthetic market.

Although CoupLink's current recommender is not identical to \onesided under sorting by dating rates, the two exhibit similar performance patterns. By incorporating ad hoc adjustments to mitigate receiver-side congestion, the current recommender slightly reduces average dates while modestly improving average effective dates and receiver-side dating probability, compared with \onesided, though the differences are subtle.

In this sense, \DA and \ECDA can be viewed as systematic and more efficient alternatives to the adjustments embedded in the current recommender. Both mechanisms intentionally trade off a reduction in the total number of dates for gains in average effective dates and receiver-side dating probability, thereby more effectively alleviating congestion and improving higher-quality matching outcomes. \DA improves congestion-adjusted metrics by imposing an upper bound on the headcount of proposers with whom a receiver can be recommend-matched, while \ECDA does so by constraining the predicted total number of likes received or dates formed. By preventing a small subset of receivers from monopolizing dating opportunities, both mechanisms mitigate congestion and improve congestion-adjusted outcomes.

\begin{figure}[tb]
  \begin{center}
       \begin{minipage}[t]{0.47\columnwidth}
            \centering
            \includegraphics[width=1.0\columnwidth]{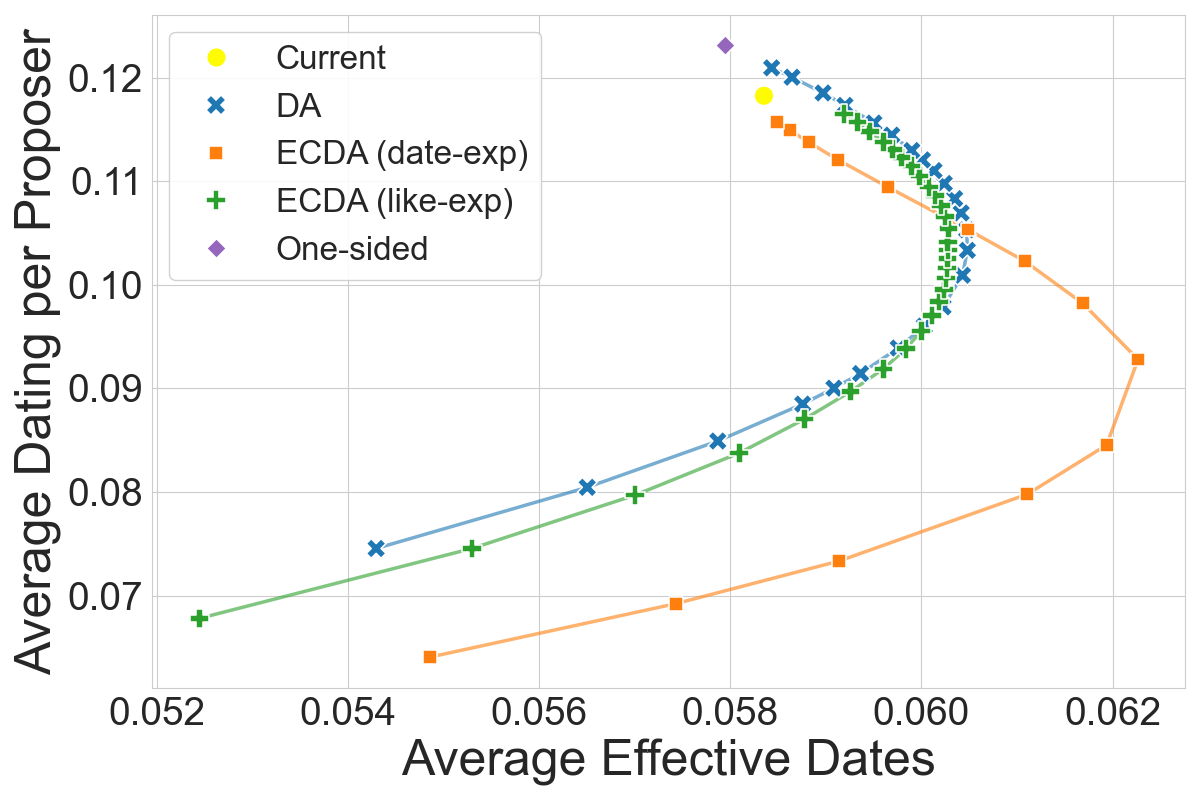}
            \subcaption{Average Effective Dates vs Average Dates}
            \label{fig:sim_effective_dates_vs_dates}
        \end{minipage}
        \hfill
        \begin{minipage}[t]{0.47\columnwidth}
            \centering
            \includegraphics[width=1.0\columnwidth]{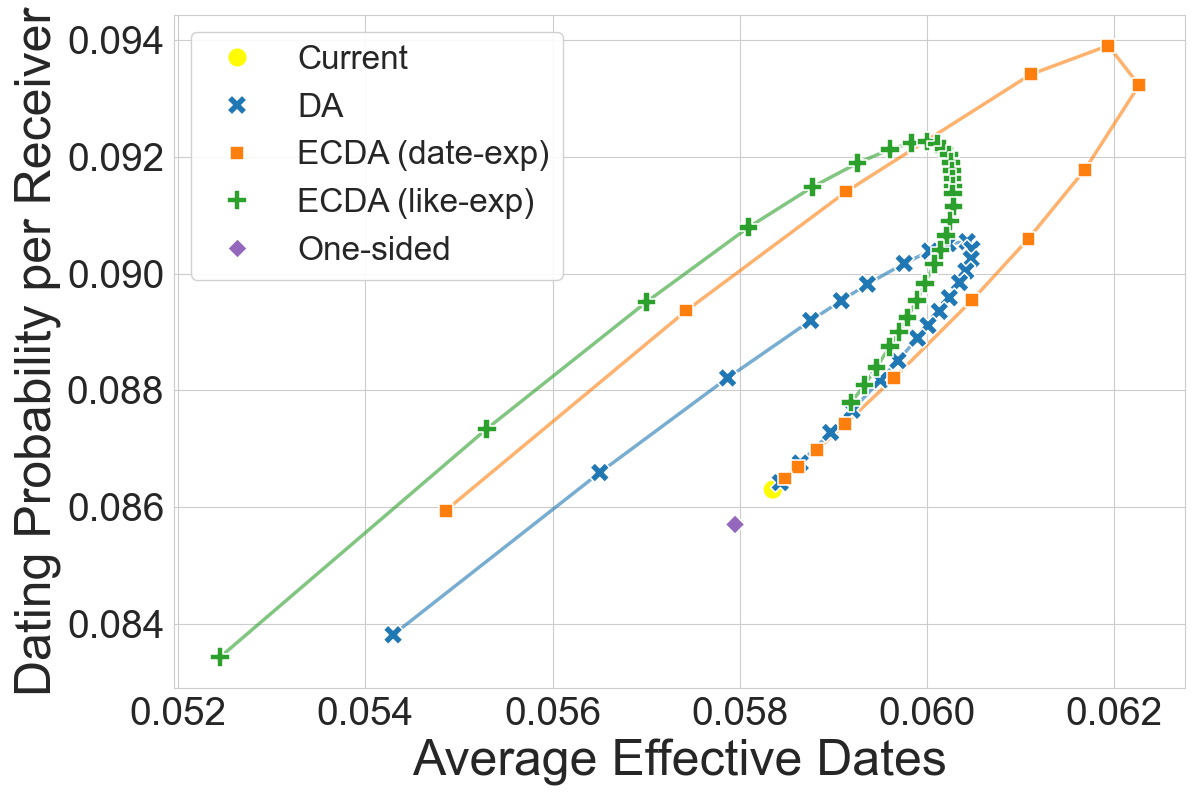}
            \subcaption{Average Effective Dates vs Dating Probability per Receiver}
            \label{fig:sim_effective_dates_vs_dating_prob}
        \end{minipage}
  \end{center}
  \caption{Performance Comparison of Different Integrators}
  \label{fig:sim_date_performance_plot}
  \begin{flushleft}\footnotesize
  Note: Both \DA and \ECDA use sorting by dating rates. For \DA, capacity ranges from 30 to 1000, with coarser grids at higher values. For \ECDA, capacity is varied based on exposure. Like exposure ranges from 4.0 to 60.0 with non-uniform step sizes, while date exposure ranges from 0.4 to 8.0. 
  \end{flushleft}
\end{figure}

Figure~\ref{fig:sim_date_performance_plot} presents scatter plots illustrating the relationship between dates, effective dates, and dating probabilities per receiver achieved by various integrators. Each point corresponds to a capacity setting; curves trace capacity sweeps. Figure~\ref{fig:sim_effective_dates_vs_dates} shows that average dates and average effective dates are positively related, but these two are not perfectly aligned. The figure plots each integrator's performance in the two-dimensional space defined by average effective dates on the horizontal axis and average dates on the vertical axis. Since higher values of both metrics are desirable, integrators that appear closer to the upper-right region of the figure can be considered to perform better overall.

The highest performance in terms of average dates is achieved by \onesided. For both \DA and \ECDA, average dates increase as receiver capacity is relaxed and converge to the performance of \onesided, which corresponds to the case of infinite receiver capacity. In contrast, the response of average effective dates to increases in receiver capacity is nonmonotonic. When receiver capacity is too small, effective dates are limited because promising proposer--receiver pairs that are likely to lead to dates cannot be included in the recommendation matching. When receiver capacity is too large, however, severe congestion arises on the top receivers, which in turn reduces the number of effective dates. As a result, across \DA, \ECDA with like exposure, and \ECDA with date exposure, effective dates are maximized at an intermediate level of receiver capacity and exceed the level achieved by \onesided.

Figure~\ref{fig:sim_effective_dates_vs_dating_prob} replaces average dates on the vertical axis with the dating probability per receiver. Effective dates capture proposer-side congestion-adjusted value, while the dating probability per receiver captures receiver-side reach. Therefore, these two metrics are important for CoupLink's platform design. Both \DA and \ECDA can achieve favorable outcomes on these dimensions by choosing an intermediate level of receiver capacity, which balances the increase in dates against the mitigation of congestion. Among the integrator designs we consider, the superior performance of \ECDA with date exposure is particularly pronounced. Congestion on the receiver side is most directly driven by dates. Thus, constraining dates is more effective than constraining likes or headcounts. Figure~\ref{fig:compare_quad} presented in Appendix~\ref{sec:simulations_with_synthetic_markets} demonstrates that \ECDA provides more equitable opportunities for being liked and dating much more effectively than \DA.

The configurations that achieve the best average effective dates (which is also approximately the best for the dating probability on the receiver side) among those tested are as follows. For \DA, effective dates are maximized at a receiver capacity of 140, which yields an average of 104.9 recommended matches per receiver. For \ECDA with like exposure, the optimal capacity is 22.0, which yields an average of 3.308 likes per receiver (as reported in Table~\ref{tab:actual_one_sided_da_ecda}). For \ECDA with date exposure, the highest performance is achieved at a capacity of 1.5, which yields an average of 0.150 dates per receiver (obtained by converting per-proposer average dates using the proposer--receiver population ratio in this simulation sample). In all three cases, the optimal capacity is relatively loose: recommendations for most receivers are effectively unconstrained, while only receivers with exceptionally high dating rates face binding capacity limits. This suggests that performance gains primarily come from mitigating excessive congestion on a small subset of top receivers, rather than from broadly restricting recommendation opportunities across the market.

\section{Field Experiment}

Offline simulations in Section~\ref{sec:simulation} suggest that \ECDA can substantially improve two-sided recommendations when (i) ROLs incorporate mutual interest (dating rates) rather than the like rate alone, and (ii) receiver-side exposure is controlled via a constraint imposed on the sum of dating rates. These properties are especially relevant for CoupLink, where a ``date'' (mutual like) is an intermediate milestone and where downstream frictions (messaging, scheduling, cognitive load) may generate additional constraints not captured by a pure matching count.

While these simulation results are informative, their validity relies on predicted behaviors and abstracts from equilibrium responses that may arise in live platforms. In practice, recommendations could affect user behavior in various ways, and changes in exposure can generate spillovers across users on both sides of the market. 

To assess whether the mechanisms identified in the simulation operate in a real-world environment, we conduct a large-scale geographic rollout that enables causal evaluation of \ECDA at production scale. We design the experiment to (i) minimize cross-user interference inherent in two-sided platforms, (ii) separately measure outcomes based on model predictions and realized behavioral responses, and (iii) measure not only matching outcomes but also post-match engagement.

\subsection{Experimental Setting}

\subsubsection{Experimental Setup}
We implemented a geographically segmented deployment and evaluated its impact using a difference-in-differences (DID) design. The treatment area was the \emph{Kanto} area, and the control area was the \emph{Kansai-Tokai} area. Note that Kanto, Kansai, and Tokai are the most populous areas of Japan. We chose an area-level rollout instead of user-level analysis because recommendation changes can generate spillovers: treated users may redirect likes toward candidates that control users would otherwise have seen, and receivers' experiences can change as the composition of inbound likes shifts. Group-level assignment helped limit within-market contamination relative to a fully randomized user-level experiment.\footnote{Two-sided recommendation affects exposure and attention, which can generate equilibrium effects: changing ``who is shown to whom'' shifts like flows, which in turn changes who receives attention and who appears responsive in the data. If treated and control users operate in the same local pool, control outcomes can be contaminated by treated users' behavior, biasing estimates (e.g., through congestion and feedback effects, as discussed by \citet{rios2023improving,manshadi2023redesigning}). A geographic rollout is a practical way to reduce such interference while remaining operationally simple.}

The first rollout of recommendations began on January 13, 2026, and ran for two weeks through January 26, 2026. During this window, the default recommendation logic for sessions originating in Kanto was switched from the status quo to our proposed \ECDA with sorting by dating rates and date exposure with a capacity of 1.5. Formally, treatment was applied whenever the user's session location (latitude/longitude recorded at the time the AI recommendation feed is requested) was mapped to a prefecture in Kanto. The control area, Kansai-Tokai, continued to operate under the status quo recommendation logic throughout the same period.\footnote{Concretely, we operationalize Kanto as Tokyo, Kanagawa, Chiba, Saitama, Ibaraki, Tochigi, and Gunma, and Kansai-Tokai as Osaka, Kyoto, Hyogo, Nara, Shiga, Wakayama, Aichi, Gifu, Shizuoka, and Mie.} The two areas are geographically distant and have a large population of male and female users on the app. Users are not informed about which recommendation logic is active in their area.

\paragraph{Treatment logic}
Based on the simulation results (in particular, the one illustrated in Figure~\ref{fig:sim_effective_dates_vs_dating_prob}), the treatment adopted \ECDA with sorting by dating rates and date exposure with a capacity of 1.5, which delivers the highest performance in terms of average effective dates. The capacity constraint limited the (predicted) expected number of dates per receiver per day to at most 1.5. As discussed in Section~\ref{subsec:simulation_results}, this upper bound is loose, and thus is binding only for a small subset of receivers who attract exceptionally many recommendations under the current recommender. In other words, most receivers did not experience a reduction in the recommendation opportunities. As a result, the introduction of this treatment had a low operational risk from a business perspective.

\subsubsection{Outcome}\label{subsubsec:experiment_outcome}
Because the platform values both market-level efficiency and equitable access to opportunities, we evaluate the rollout using a layered set of outcome measures. All outcomes are computed at the area $\times$ day level, consistent with CoupLink's daily recommendation refresh cycle.

First, we consider predicted outcomes at recommendation, which are constructed from model predictions at the time recommendations are generated. They quantify expected matching opportunities and congestion prior to user actions, allowing for a direct comparison with the simulation results.

Second, we examine realized outcomes, which are based on actual user behavior observed within two weeks after the recommendation. Unlike predicted outcomes, these measures reflect how changes in exposure translate into observed matching behavior in reality and can only be obtained through a live deployment.

Finally, we track post-engagement outcomes, which capture user interaction following a realized date. These include the average number of messages sent and the probability of messaging on both the proposer and receiver sides, measuring whether improvements in exposure and match formation propagate to downstream engagement.

Formal constructions of all outcome variables are provided in Appendix~\ref{sec:outcome_definitions}.

\subsubsection{Difference-in-Differences Methodology}

We estimate intention-to-treat (ITT) effects using a DID design that compares outcomes in geographical areas assigned to the treatment area (Kanto) with those in geographical areas assigned to the control area (Kansai-Tokai). Compared with a finer geographic split, Kanto and Kansai-Tokai are separated by substantial distance, reducing boundary contamination from users who physically cross the treated/control border within short time horizons. Let $a$ index geographical areas and $t$ index calendar days. For an outcome $Y_{ta}$ observed at the area-day level, our baseline specification is
\begin{equation}
\label{eq:did-baseline}
Y_{ta}
= \alpha_a + \delta_t
+ \beta_{\text{ITT}}\mathbf{1}\{a=\text{Kanto}\}\times\mathbf{1}\{t\in\mathcal{T}_{\text{deploy}}\}
+ \varepsilon_{ta},
\end{equation}
where $\alpha_a$ are area fixed effects and $\delta_t$ are day fixed effects capturing common shocks (seasonality, holidays, marketing, etc.), and $\mathcal{T}_{\text{deploy}}$ is the set of treatment days.\footnote{The estimand is the assignment-based ITT effect of deploying \ECDA in Kanto during the rollout window. Because treatment variation is at the area level with one treated and one control macro-unit, conventional large-cluster asymptotics are not available. We therefore use the pooled DID coefficients primarily as post-period summary effect measures and interpret associated standard errors and significance stars as suggestive rather than definitive.}


\paragraph{Identification Assumptions.}
Identification of $\beta_{\text{ITT}}$ in Equation~\eqref{eq:did-baseline} relies on the following assumptions:
(i) \textit{parallel trends}: absent \ECDA, treated and control areas would have evolved similarly;
(ii) \textit{no anticipation}: outcomes did not respond before $t_0$;
(iii) \textit{stable composition}: the intervention did not differentially change the population at risk across areas (e.g., the stock of active users);
and (iv) \textit{limited interference}: spillovers between Kanto and Kansai-Tokai were small.
We assess (i) by inspecting pre-trends in the main outcomes in Figure~\ref{fg:event_study}. Assumption (ii) is supported operationally because the deployment was not announced.
Assumptions (iii) and (iv) are supported by the fact that user mobility across areas was negligible: only 0.003\% of users ever switched areas during the study window.

\subsection{Results}

\subsubsection{Event Study Evidence}

\begin{figure}[tp]
    \begin{center}
    \begin{minipage}[b]{0.47\columnwidth}
        \centering
        \includegraphics[width=1.0\columnwidth]{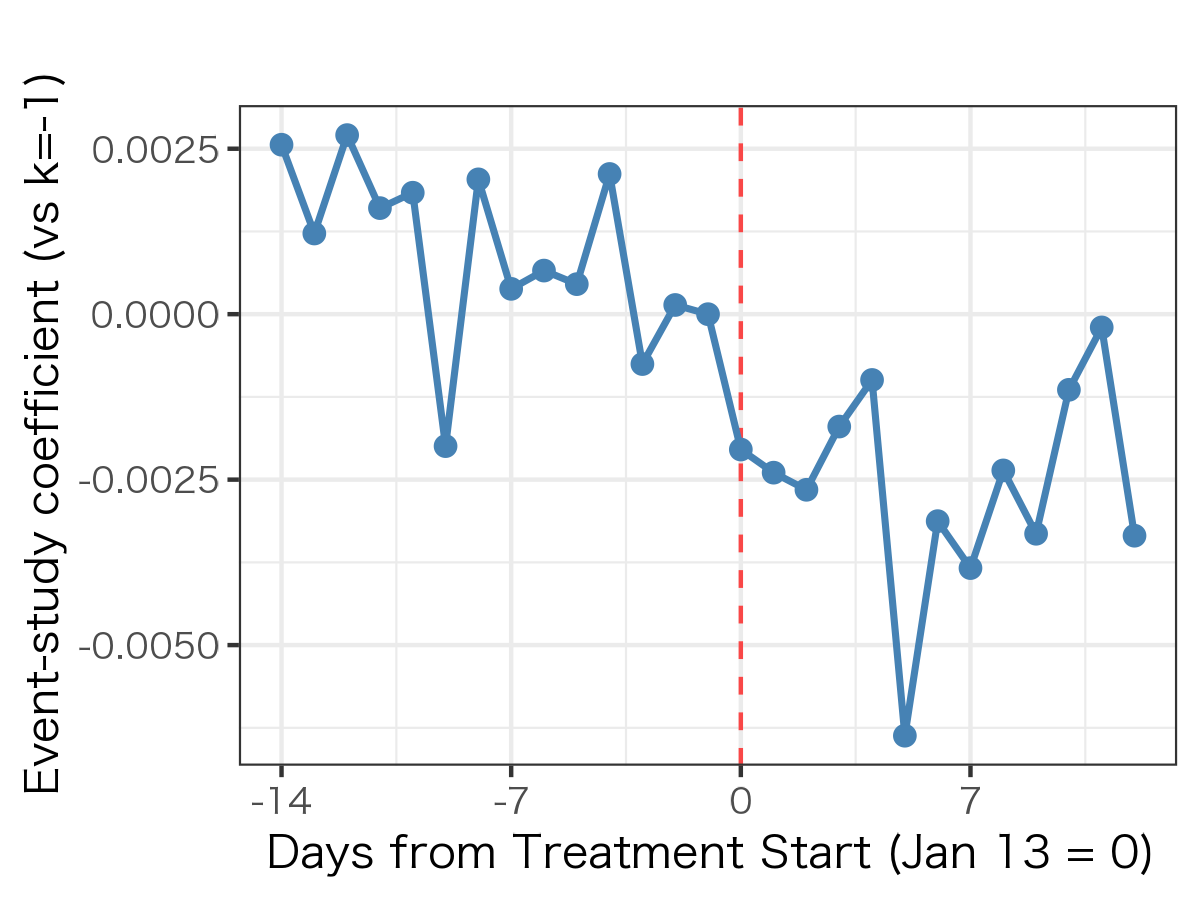}
        \subcaption{Avg Dates (Predicted)}
    \end{minipage}
    \hfill
    \begin{minipage}[b]{0.47\columnwidth}
        \centering
        \includegraphics[width=1.0\columnwidth]{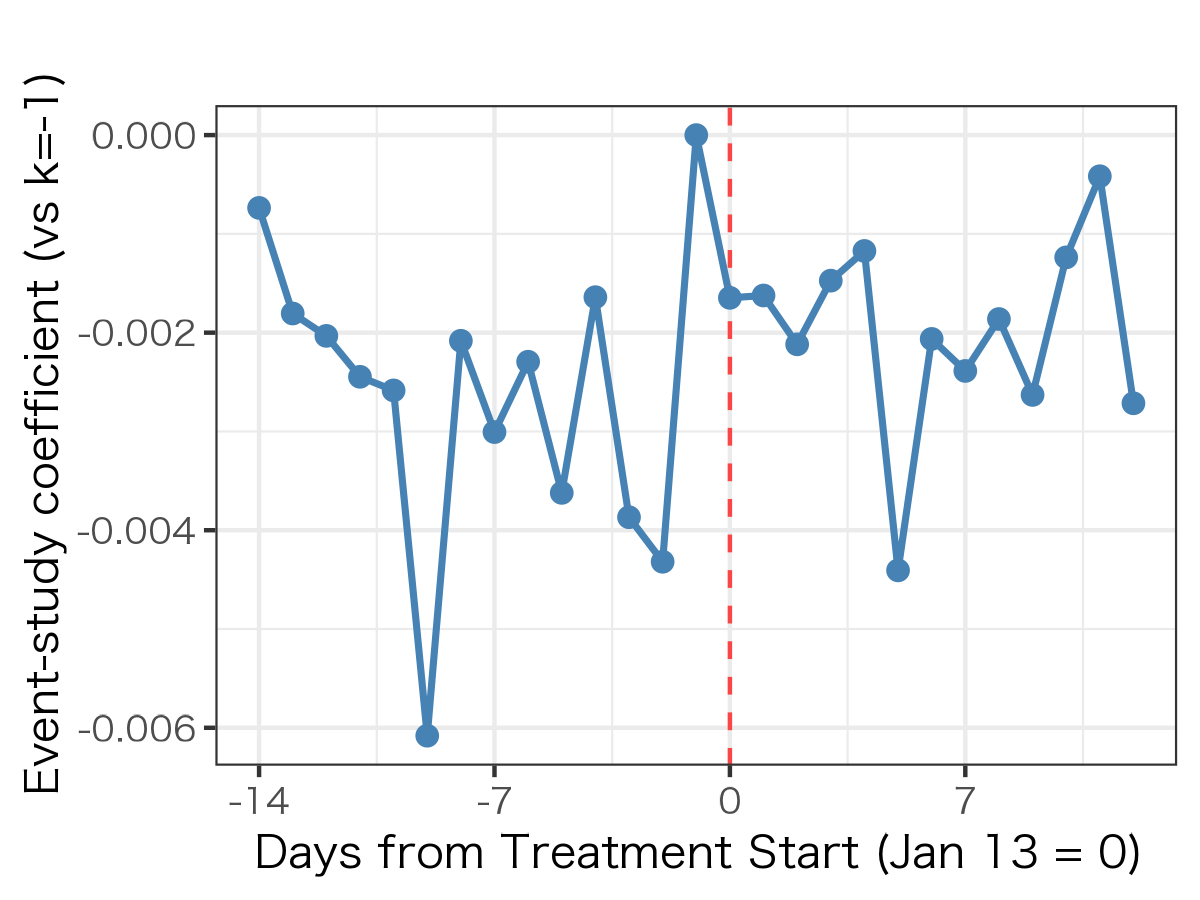}
        \subcaption{Avg Effective Dates (Predicted)}
    \end{minipage}\\
    \begin{minipage}[b]{0.47\columnwidth}
        \centering
        \includegraphics[width=1.0\columnwidth]{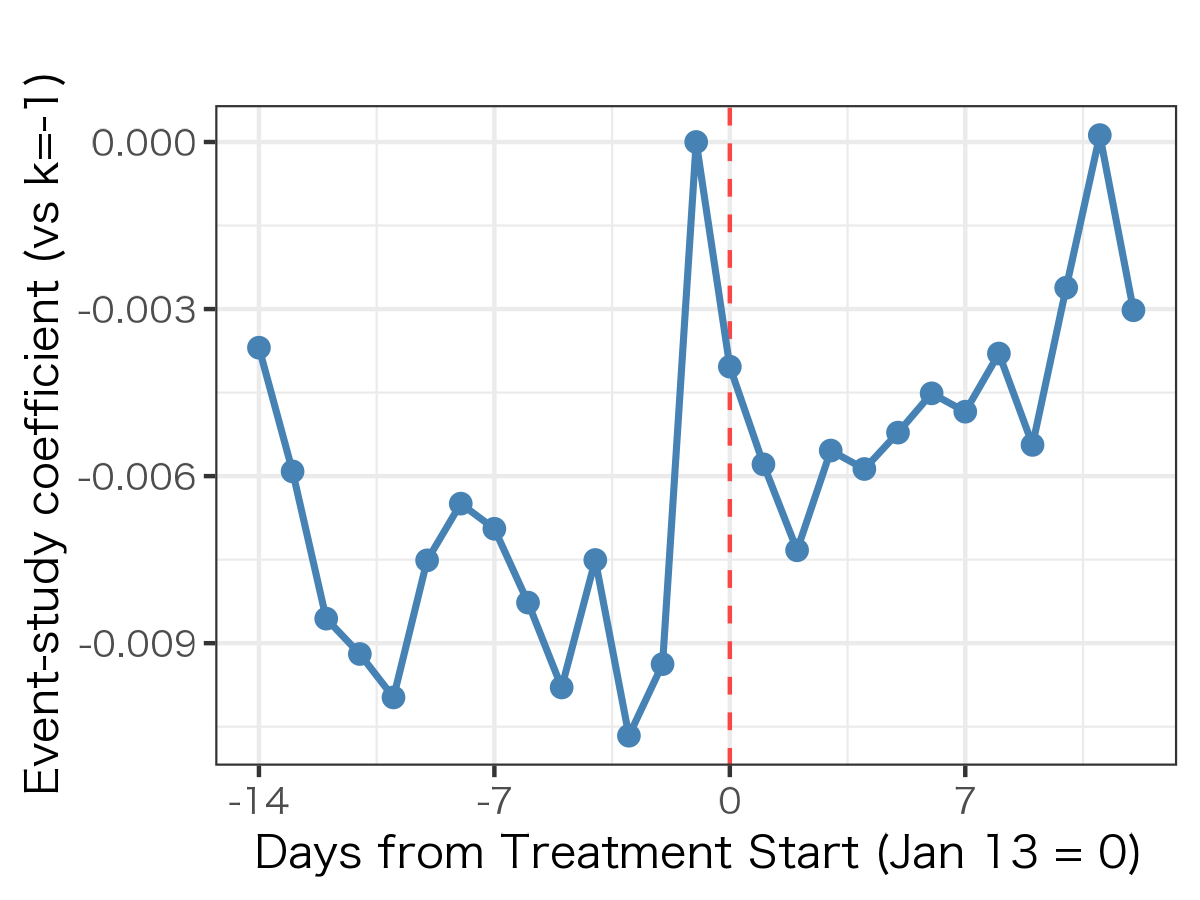}
        \subcaption{Dating Prob (Receiver, Predicted)}
    \end{minipage}
    \hfill
    \begin{minipage}[b]{0.47\columnwidth}
        \centering
        \includegraphics[width=1.0\columnwidth]{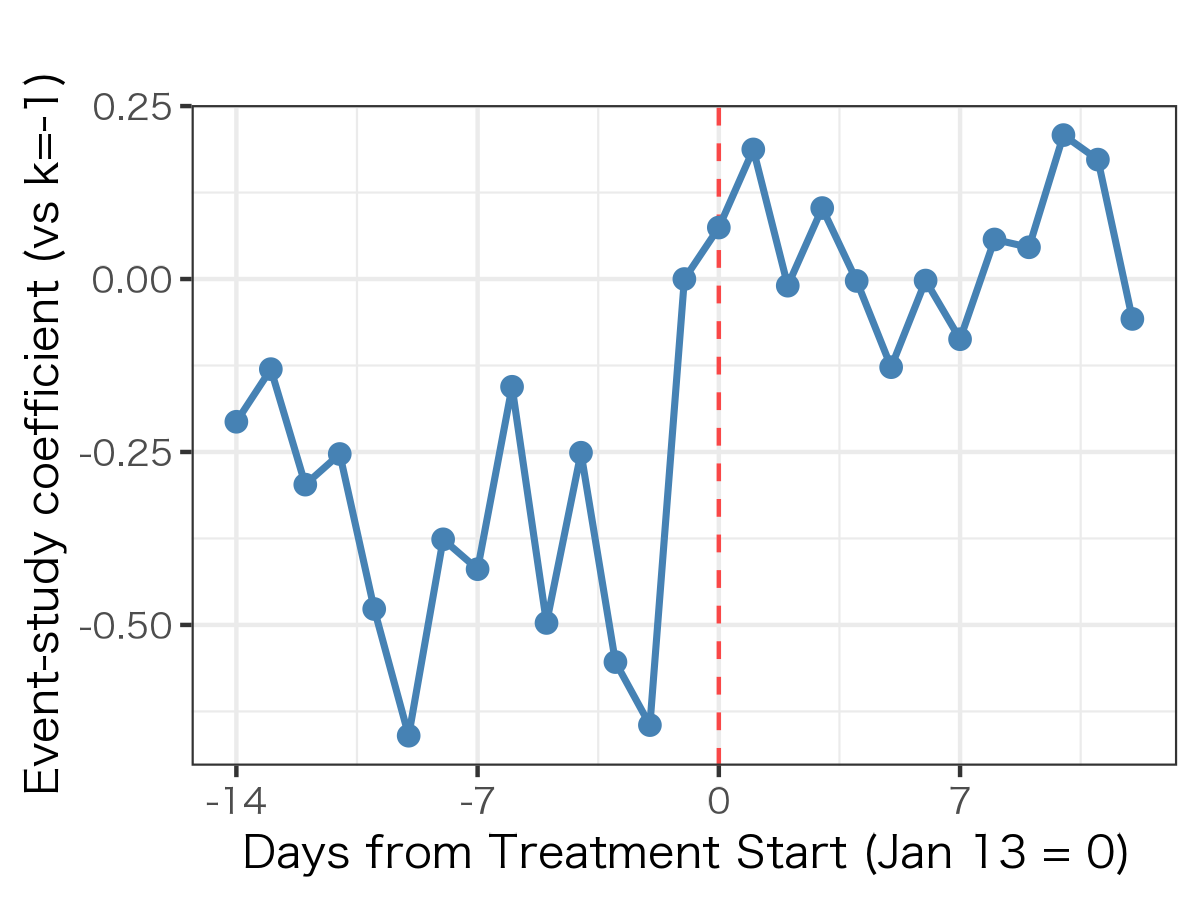}
        \subcaption{Avg Likes (Receiver, Predicted)}
    \end{minipage}\\
    \begin{minipage}[b]{0.47\columnwidth}
        \centering
        \includegraphics[width=1.0\columnwidth]{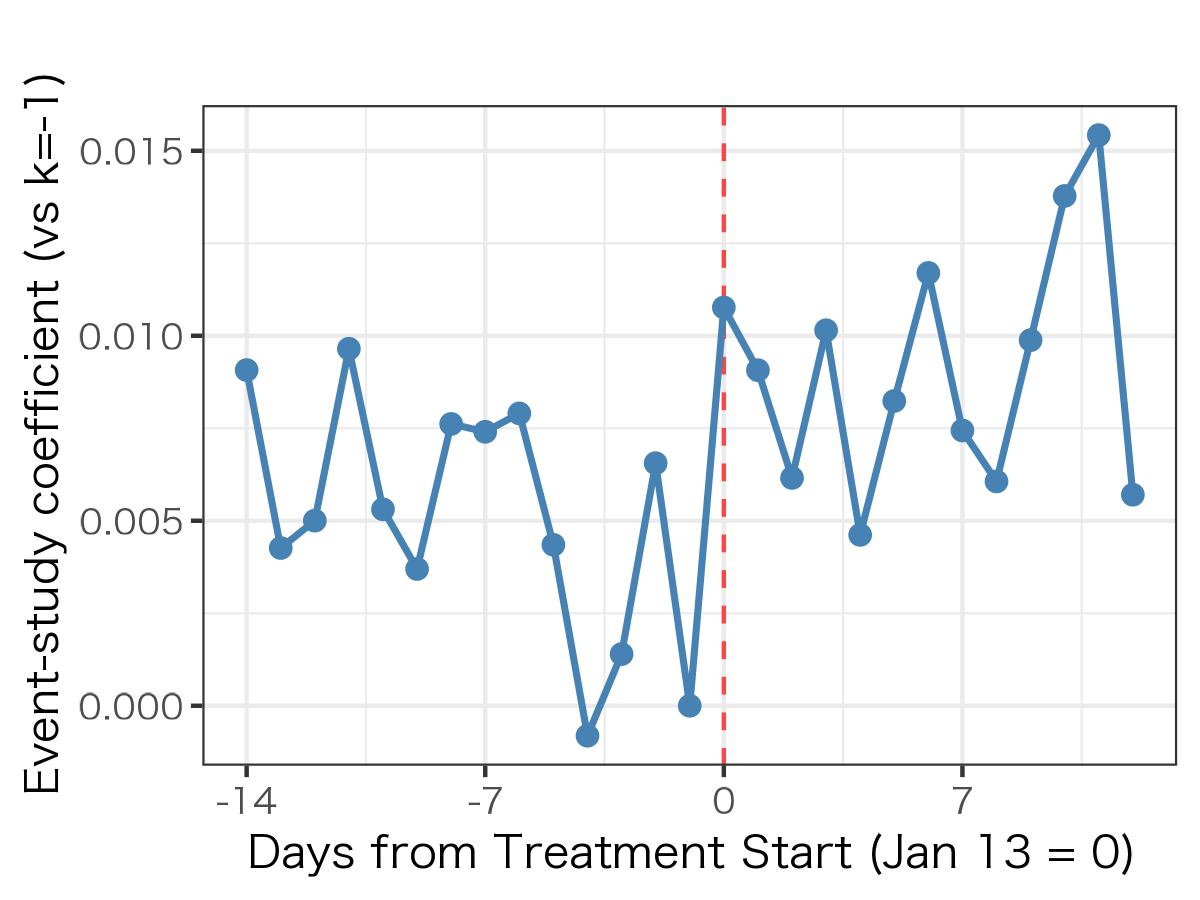}
        \subcaption{Avg Effective Dates (Realized)}
    \end{minipage}
    \hfill
    \begin{minipage}[b]{0.47\columnwidth}
        \centering
        \includegraphics[width=1.0\columnwidth]{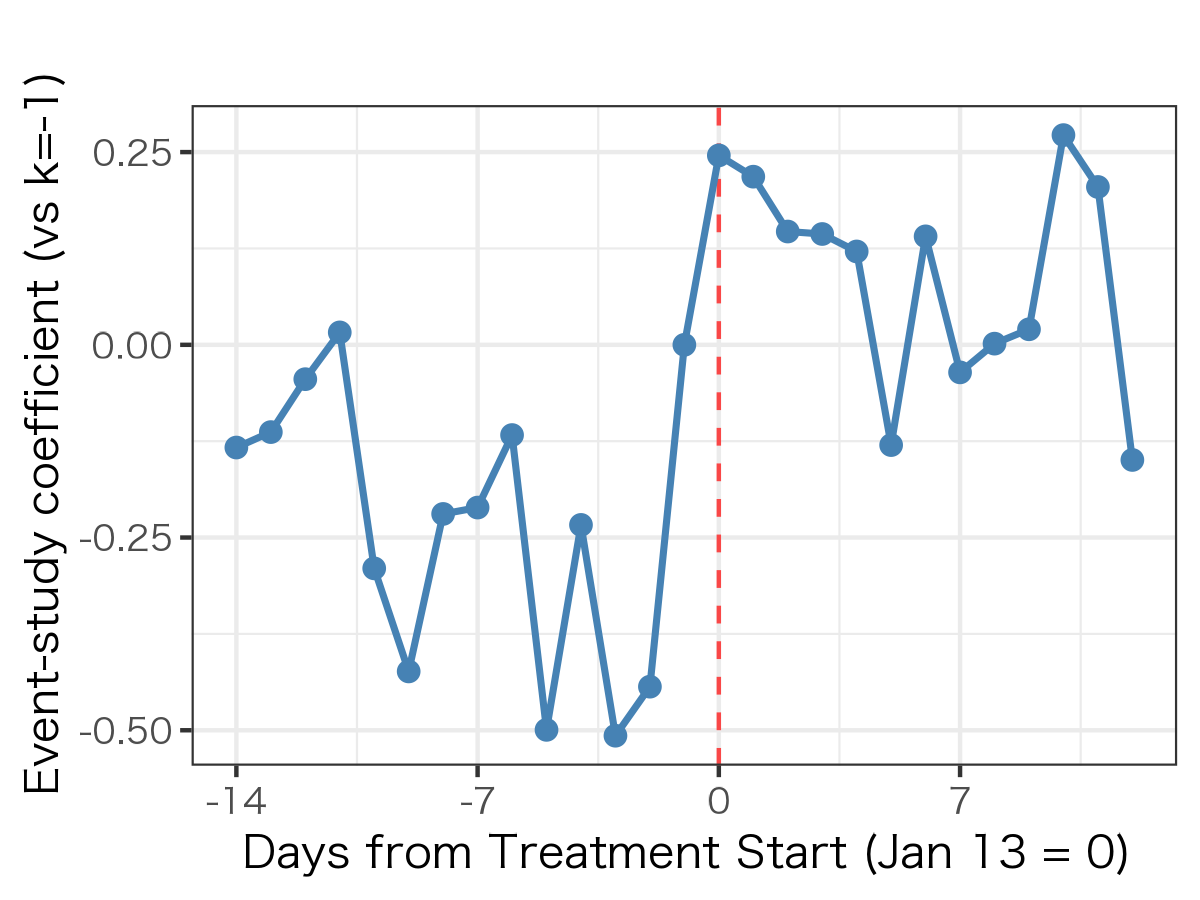}
        \subcaption{Avg Likes (Receiver, Realized)}
    \end{minipage}
\caption{Event-Study Coefficients Relative to $k=-1$}
\label{fg:event_study}
\end{center}
\begin{flushleft}\footnotesize
  Note: Each panel plots coefficients on treatment $\times$ event-time from an event-study regression with area and day fixed effects, normalized to zero at $k=-1$. The red dashed line marks treatment onset (Jan 13, 2026). Realized outcomes exclude top 0.1\% receiver-days for reasons discussed in Section \ref{sec:did_estimates}. Confidence intervals are omitted because standard inference is unreliable with only two geographic clusters.
  \end{flushleft}
\end{figure}

Figure~\ref{fg:event_study} reports coefficients from an event-study regression that interacts the treatment-area indicator with event-time dummies around the deployment date, controlling for area and day fixed effects. The coefficients are normalized to zero at $k=-1$ (the last pre-treatment day) and represent treatment--control differences relative to that baseline. Because the event study is identified using only two geographic units, standard errors are not reliable; we therefore treat the event-study plots as descriptive and use the pooled DID estimates in Table~\ref{tb:date_area_did_results_with_fe} as post-period summaries. 

After treatment onset ($k=0$), predicted average dates decline while predicted effective dates and receiver-side dating probability increase, consistent with the reallocation of recommendations away from congested receivers. The same qualitative patterns are also visible in realized average effective dates and average likes.

\subsubsection{Difference-in-Differences Estimates}\label{sec:did_estimates}


\begin{table}[t]
  \begin{center}
      \caption{Difference-in-Differences Estimates}
      \label{tb:date_area_did_results_with_fe}
      \subfloat[Predicted Outcomes at Recommendation]{
\begin{tabular}[t]{lccccc}
\toprule
  & (1) & (2) & (3) & (4) & (5)\\
\midrule
Dependent Var & Avg Dates & Avg Eff Dates & DatProb(P) & DatProb(R) & AvgLikes(R)\\
Treat $\times$ Post & -0.003*** & 0.001** & -0.001*** & 0.004*** & 0.415***\\
 & (0.001) & (0.000) & (0.000) & (0.001) & (0.054)\\
\midrule
Num.Obs. & 56 & 56 & 56 & 56 & 56\\
R2 Adj. & 0.904 & 0.787 & 0.884 & 0.974 & 0.969\\
\bottomrule
\end{tabular}
}\\
      \bigskip
      \subfloat[Realized Outcomes in Two Weeks]{
\begin{tabular}[t]{lccccc}
\toprule
  & (1) & (2) & (3) & (4) & (5)\\
\midrule
Dependent Var & Avg Dates & Avg Eff Dates & DatProb(P) & DatProb(R) & AvgLikes(R)\\
Treat $\times$ Post & -0.002 & 0.002 & -0.001 & 0.003 & 0.264***\\
 & (0.002) & (0.001) & (0.001) & (0.002) & (0.059)\\
\midrule
Num.Obs. & 56 & 56 & 56 & 56 & 56\\
R2 Adj. & 0.469 & 0.239 & 0.346 & 0.547 & 0.930\\
\bottomrule
\end{tabular}
}\\
      \bigskip
      \subfloat[Realized Post-engagement in Two Weeks]{
\begin{tabular}[t]{lcccc}
\toprule
  & (1) & (2) & (3) & (4)\\
\midrule
Dependent Var & AvgMsg(P) & AvgMsg(R) & MsgProb(P) & MsgProb(R)\\
Treat $\times$ Post & -0.002 & -0.004* & -0.001* & -0.001\\
 & (0.001) & (0.002) & (0.001) & (0.001)\\
\midrule
Num.Obs. & 56 & 56 & 56 & 56\\
R2 Adj. & 0.266 & 0.498 & 0.332 & 0.497\\
\bottomrule
\end{tabular}
}\\
  \end{center}\footnotesize
  \textit{Note}: (P) = Proposer Side, (R) = Receiver Side. Area and date fixed effects are included. Standard errors in parentheses. * $p<0.10$, ** $p<0.05$, *** $p<0.01$. Outcomes are defined in Section \ref{sec:performance_indicator}.
\end{table}

Table~\ref{tb:date_area_did_results_with_fe} reports DID estimates.
Panel (a) shows that the expected-outcome measures move in the direction predicted by the simulation and are statistically significant. In other words, up to the recommendation stage, the rollout affects the marketplace exactly as designed: it improves fairness in exposure and expected engagement, consistent with the simulation results in Section \ref{sec:simulation}. 

Panels (b)–(c) jointly test whether the recommender-driven improvements in prediction-based metrics translate into realized actions. In Panel (b), the treatment increases the volume of likes on the receiver side, indicating that the algorithm successfully improves upstream engagement and matching opportunities. However, the effect does not carry through to the final conversion margin: the estimated change in the realized date probability is small and statistically indistinguishable from zero. Panel (c) further shows an apparent tension with the upstream gains: post-engagement outcomes move in the opposite direction, with negative and in some cases statistically significant effects on the average number of messages and the probability of messaging, suggesting that even when more matches are initiated, subsequent interaction intensity does not increase and may even decline in the short run.

\begin{table}[t]
  \begin{center}
      \caption{Difference-in-Differences Estimates Excluding Top 0.1\% Receivers}
      \label{tb:date_area_did_results_with_fe_no_outlier}
      \subfloat[Realized Outcomes in Two Weeks excluding top 0.1\% receivers]{
\begin{tabular}[t]{lccccc}
\toprule
  & (1) & (2) & (3) & (4) & (5)\\
\midrule
Dependent Var & Avg Dates & Avg Eff Dates & DatProb(P) & DatProb(R) & AvgLikes(R)\\
Treat $\times$ Post & 0.003 & 0.003** & 0.002* & 0.005* & 0.334***\\
 & (0.002) & (0.001) & (0.001) & (0.002) & (0.058)\\
\midrule
Num.Obs. & 56 & 56 & 56 & 56 & 56\\
R2 Adj. & 0.347 & 0.217 & 0.347 & 0.515 & 0.929\\
\bottomrule
\end{tabular}
}\\
      \bigskip
      \subfloat[Realized Post-engagement in Two Weeks excluding top 0.1\% receivers]{
\begin{tabular}[t]{lcccc}
\toprule
  & (1) & (2) & (3) & (4)\\
\midrule
Dependent Var & AvgMsg(P) & AvgMsg(R) & MsgProb(P) & MsgProb(R)\\
Treat $\times$ Post & 0.000 & 0.000 & 0.000 & 0.000\\
 & (0.001) & (0.002) & (0.001) & (0.001)\\
\midrule
Num.Obs. & 56 & 56 & 56 & 56\\
R2 Adj. & 0.164 & 0.416 & 0.162 & 0.449\\
\bottomrule
\end{tabular}
}
  \end{center}\footnotesize
  \textit{Note}: (P) = Proposer Side, (R) = Receiver Side. Area and date fixed effects are included. Standard errors in parentheses. * $p<0.10$, ** $p<0.05$, *** $p<0.01$. Outcomes are defined in Section \ref{sec:performance_indicator}.
\end{table}

To probe whether this discrepancy is driven by a small set of receivers with exceptionally high dating rates, Table~\ref{tb:date_area_did_results_with_fe_no_outlier} repeats the analysis after excluding the top 0.1\% of receiver-day observations.\footnote{The qualitative conclusions are robust to alternative thresholds: we obtain similar results when excluding the top 0.5\% or 1\% of receiver-day observations. We report the 0.1\% threshold as it is the most conservative and minimal exclusion that yields coherent treatment effects.} This type of tail-trimming exercise is common in matching marketplace evaluations \citep[e.g.,][]{manshadi2023redesigning}. Importantly, we do not exclude these top receiver-day observations in the baseline DID specification: Panels (b) and (c) in Table~\ref{tb:date_area_did_results_with_fe} report the full-sample effects, while Table~\ref{tb:date_area_did_results_with_fe_no_outlier} uses trimming only as a secondary decomposition to assess whether extreme right-tail congestion can distort aggregate averages. The results become more coherent with the mechanism: Panel (a) shows statistically significant positive effects on average effective dates, dating probabilities, and average likes, while the coefficient on average dates is positive but not statistically significant; Panel (b) shows no meaningful effect on post-engagement outcomes. In other words, once we remove the top receivers who likely face severe congestion and crowding-out, the recommender improves most early-stage outcomes for the remaining 99.9\% of receivers, but does not materially change post-date engagement. This pattern suggests that the recommender improves outcomes up to the formation of mutual likes, but that further improvements in substantive match quality---such as sustained messaging and interaction---are not well captured in our data and may require separate tuning or design changes beyond the ranking algorithm.


\begin{figure}[tp]
    \begin{center}
    \begin{minipage}[b]{0.47\columnwidth}
        \centering
        \includegraphics[width=1.0\columnwidth]{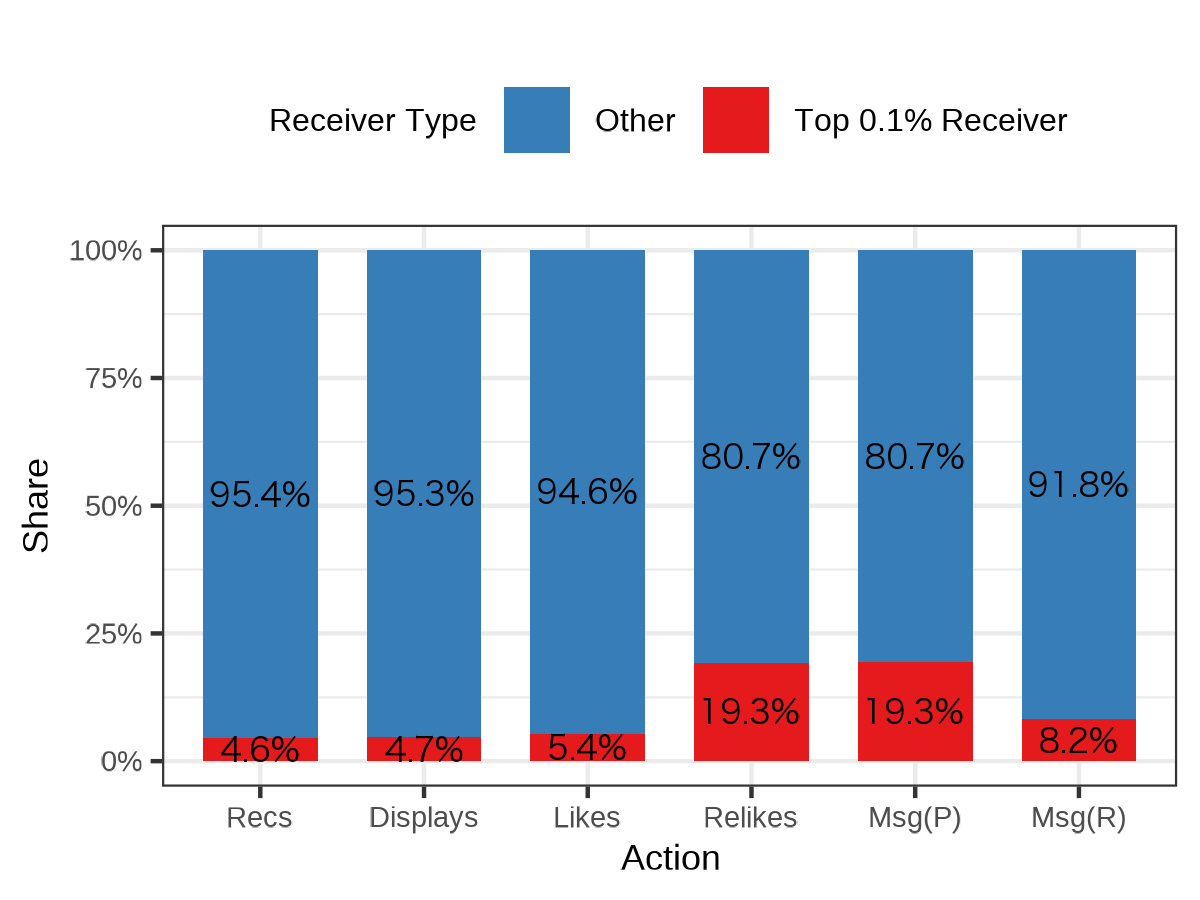}
        \subcaption{Share of Activity by Receiver Type across Funnel Stages}
    \end{minipage}
    \hfill
    \begin{minipage}[b]{0.47\columnwidth}
        \centering
        \includegraphics[width=1.0\columnwidth]{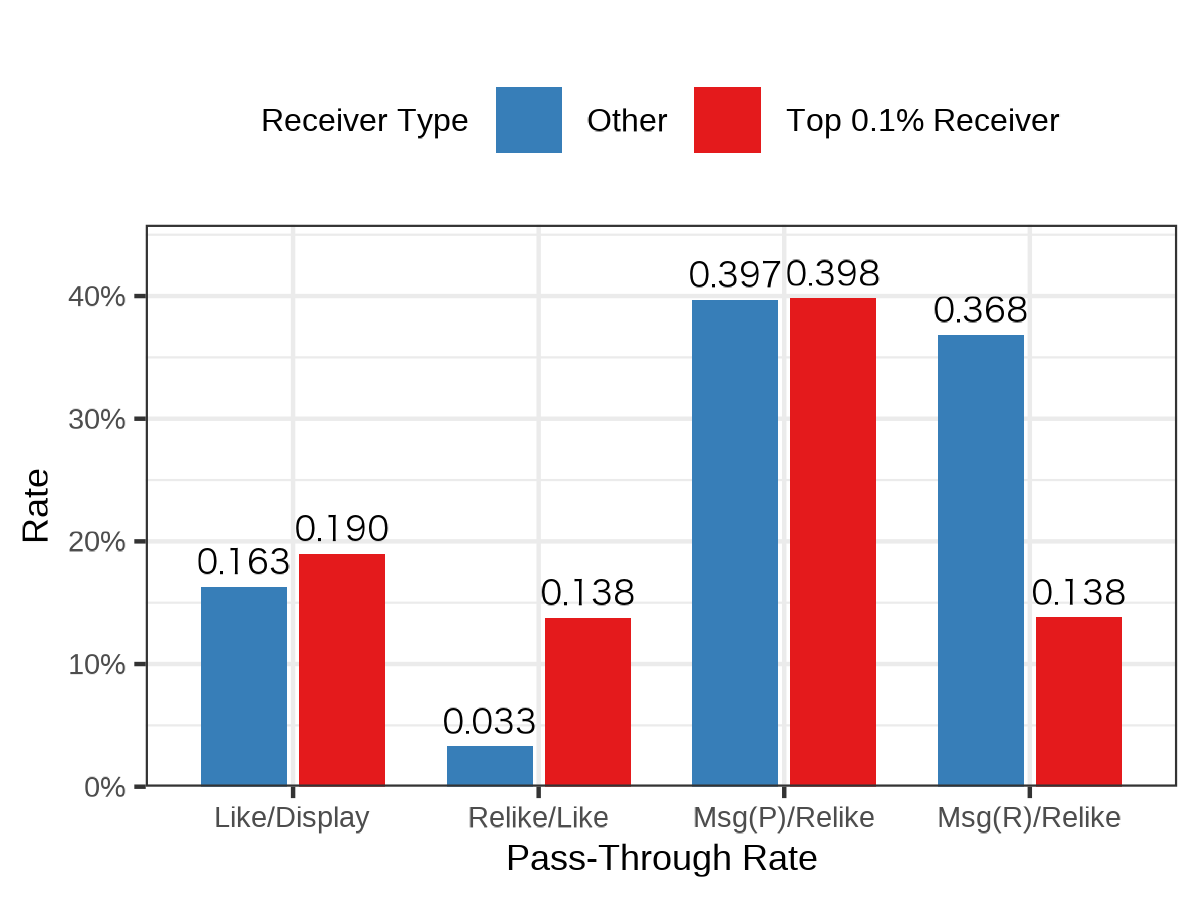}
        \subcaption{Pass-Through Rates across Stages of the Matching Funnel by Receiver Type}
    \end{minipage}
\caption{Activities and Pass-Through Rates of Top 0.1\% Receivers and Others (Kanto Pre-Period)}
\label{fg:outlier_receiver_passthrough_kanto_pre}
\end{center}
\end{figure}

Figure \ref{fg:outlier_receiver_passthrough_kanto_pre} illustrates how a very small group of receivers concentrates attention and downstream interactions in the market. Panel (a) shows that the top 0.1\% of receivers attract a disproportionate share of early-stage exposure: in the Kanto pre-period, they receive 4.6\% of all recommendations (and similarly elevated shares of displays and likes), indicating substantial congestion on the receiver side even before any treatment. 

Panel (b) then examines pass-through along the matching funnel. Top receivers are characterized by two striking features. First, they exhibit exceptionally high relike rates, which mechanically attracts a disproportionate volume of recommendations and likes relative to other receivers. Second, despite this strong early-stage responsiveness, they exhibit exceptionally low post-engagement activity. Post-engagement pass-through deteriorates sharply: message exchange conditional on the receiver's relike (i.e., date formation) is roughly one-third of that for other receivers. Taken together, these patterns indicate that extreme receiver-side congestion arises from the combination of unusually high responsiveness at the relike stage and weak follow-through in communication, allowing top receivers to dominate matches through the message-sending stage, yet preventing many of those matches from translating into substantive interaction.

Our algorithm is designed precisely to mitigate this early-stage congestion by reallocating recommendation exposure away from the extreme right tail and toward the broader receiver population, thereby improving match formation for the remaining 99.9\% of users.

\section{Conclusions}
We study integrator design for two-sided recommendation in online dating, where successful matches require mutual interest, activity, and responsiveness on both sides. Using production-grade predictions from CoupLink, we show that conventional one-sided ranking concentrates recommendations on a small subset of receivers, generating severe congestion and overstating efficiency when evaluated by raw dates. To address this, we introduce effective dates, a congestion-adjusted outcome metric, and propose \ECDA, which limits receiver exposure in terms of expected likes or dates rather than headcount. Across calibrated simulations and a large empirical market, \ECDA with dating-rate sorting consistently achieves higher effective dates and dating probability per receiver. Performance gains arise primarily from imposing loose exposure constraints that bind only for receivers with exceptionally high dating rates, thereby mitigating congestion without broadly restricting recommendations. A large-scale field experiment confirms that these expected effects are realized in practice. Post-engagement outcomes (messaging) do not materially change, suggesting that \ECDA improves match formation without degrading subsequent interaction quality.

From a practical standpoint, our results highlight exposure control as a central design lever in two-sided platforms. \ECDA is simple, scalable, and well-suited for production deployment: With dating-rate sorting, it admits fast greedy implementation, and its capacity parameter has a clear interpretation that allows conservative tuning. Field evidence further suggests that excessive congestion is driven by receivers with exceptionally high relike rates rather than by differences in general attractiveness or activity, implying that exposure rebalancing can simultaneously improve equity and expand opportunities for the broader user base. More broadly, the framework developed here, i.e., combining congestion-adjusted metrics with exposure-based constraints, extends naturally to other two-sided platforms, such as hiring and gig markets, where recommendations shape attention and downstream frictions make raw matches an incomplete measure of success.

\bibliographystyle{ecta}
\bibliography{references}

\newpage
\appendix

\section{Derivation of Effective Dating Rates}\label{sec:effective_dates_derivation}

We assume that whether a proposer--receiver pair forms a date is determined independently across pairs.
Conditional on forming dates, we assume that a receiver who has dated multiple proposers selects exactly one proposer uniformly at random to develop a relationship.
We refer to a date that is selected in this manner as an \emph{effective date}.

Let $\delta_{ij}^\dag$ denote the probability that proposer $i$ and receiver $j$ form an effective date when $j$ is recommended to $i$.
Let $X_{j,-i}$ be the random variable representing the number of dates that receiver $j$ forms with proposers other than $i$.
Then $X_{j,-i}$ follows a Poisson binomial distribution with success probabilities $(\delta_{kj} M_{kj})_{k \neq i}$.
An effective date between $i$ and $j$ occurs if and only if (i) $i$ and $j$ form a date, and (ii) proposer $i$ is selected among the $X_{j,-i}+1$ dated proposers. The events (i) and (ii) are assumed to be independent, the event (i) occurs with probability $\delta_{ij}M_{ij}$, and the event (ii) occurs with probability $1/(1 + X_{j,-i})$.
Therefore,
\begin{equation}
\delta_{ij}^\dag M_{ij} = \delta_{ij} M_{ij} \mathbb{E}\!\left[\frac{1}{1 + X_{j,-i}}\right].
\end{equation}

Using Le Cam's theorem for Poisson binomial distributions, $X_{j,-i}$ can be well approximated by a Poisson random variable with parameter
\begin{equation}
\mu_{j,-i} = \sum_{k \neq i} \delta_{kj} M_{kj}.
\end{equation}
The approximation error is of order $O(\sum_{k \neq i} \delta_{kj}^2 M_{kj}^2)$, which is negligible in our setting because $\delta_{kj}$ is typically very small.

If a random variable $X$ follows a Poisson distribution with parameter $\mu$, then
\begin{equation}
\mathbb{E}\left[\frac{1}{1+X}\right] = \sum_{k=0}^\infty \frac{1}{1+k}\frac{\mu^k e^{-\mu}}{k!} = \frac{1}{\mu}\sum_{k = 1}^\infty \frac{\mu^k e^{-\mu}}{k!} = \frac{1-e^{-\mu}}{\mu}.
\end{equation}
Combining these results yields the approximation
\begin{equation}
\delta_{ij}^\dag\approx\frac{1-e^{-\mu_{j,-i}}}{\mu_{j,-i}}\delta_{ij}.
\end{equation}

Finally, since $\delta_{ij} M_{ij}$ is typically much smaller than
$\mu_{j,-i}$, we approximate $\mu_{j,-i}$ by
\begin{equation}
\mu_j \coloneqq \sum_{k \in [I]} \delta_{kj} M_{kj}
= \mu_{j,-i} + \delta_{ij} M_{ij}.
\end{equation}
This leads to
\begin{equation}
\delta_{ij}^\dag \approx \frac{1-e^{-\mu_j}}{\mu_j}\delta_{ij} \eqqcolon \delta_{ij}^*.
\end{equation}

\section{Outcome Definitions}\label{sec:outcome_definitions}
Outcomes are computed at the area $\times$ day level to align with CoupLink's daily recommendation refresh cycle.

\subsection{Predicted outcomes at recommendation}
These metrics are computed from model predictions at the time recommendations are generated. Other than restricting the sets of proposers and receivers to active users, the definitions coincide with those in the simulation, presented in Section~\ref{sec:performance_indicator}.
\begin{description}
\item[Average Dates:] The average number of predicted dates per active proposer.
\item[Average Effective Dates:] The average number of predicted effective dates per active proposer.
\item[Average Likes (Receiver):] The average number of predicted likes received per receiver.
\item[Dating Probability (Proposer/Receiver):] The predicted fraction of active proposers (receivers) obtaining at least one dating opportunity.
\end{description}

Note that average dates and average likes are defined on a per-proposer and per-receiver basis, respectively. Since the corresponding totals are fixed at the market level, computing these averages using the population size on the opposite side would simply rescale the values by a constant factor and would not affect their relative magnitudes or comparisons across integrators.

\subsection{Realized outcomes}
These metrics track realized behavior observed within two weeks after the recommendation:
\begin{description}
\item[Average Dates:] The average number of realized dates (i.e., proposer--receiver pairs who send likes mutually) per active proposer.
\item[Average Effective Dates:] The average number of expected effective dates, conditional on realization of dates. Since we cannot observe which date becomes an effective date (i.e., which proposer--receiver pair develops a deeper relationship), we assume that, for each receiver and each day, one of the realized dates is selected uniformly at random to become an effective date. Under this assumption, the expected number of effective dates for proposer $i$ conditional on all date events is given by
\begin{equation}
    \sum_{j \in [J]}\frac{\mathbbm{1}\{\text{$i$ and $j$ had a date}\}}{\#\{i' \in [I]: \text{ $i'$ and $j$ had a date}\}},
\end{equation}
and average effective dates is defined as the proposer-side expectation of this quantity. Under this realized definition, the metric is mechanically linked to receiver-side dating probability because both depend on whether each receiver has at least one realized date on a given day; they differ mainly by normalization and proposer-side allocation.
\item[Average Likes (Receiver):] The average number of likes actually received per receiver.
\item[Dating Probability (Proposer/Receiver):] The fraction of active proposers (receivers) who obtained at least one realized date.
\end{description}

\subsection{Post-engagement outcomes}
These metrics measure engagement after a date is formed:
\begin{description}
\item[Average Message (Proposer/Receiver):] The average number of messages sent by proposers (receivers) within two weeks after the recommendation.
\item[Message Probability (Proposer/Receiver):] The fraction of proposers (receivers) who sent at least one message within two weeks after the recommendation.
\end{description}

\section{Proofs}\label{sec: proofs}

\subsection{Proof of Theorem~\ref{thm: DA is equivalent to greedy}}

\begin{proof}
    The output of \DA is independent of the order in which proposers make proposals. We sort proposer--receiver pairs in descending order of dating rates $\delta_{ij}$ and sequentially select proposers according to this order. By construction, when proposer $i$ is called because it is the turn of pair $(i, j)$, proposer $i$ indeed proposes receiver $j$ because (i) proposer $i$ has not been rejected by receiver $j$, and (ii) proposer $i$ has already proposed receivers preferred to $j$ (i.e., for all $k \in [J]$ such that $\delta_{ik} > \delta_{ij}$). Since proposals are made in descending order of dating rates, receivers also prefer earlier proposals. Accordingly, receivers accept any proposals until their capacity becomes full, and reject all proposals afterward. As all ``tentative'' matches are final, the output of \DA is equivalent to the greedy algorithm, which matches pairs in descending order of dating rates, $\delta_{ij}$.
\end{proof}

\subsection{Proof of Theorem~\ref{thm: ECDA is equivalent to greedy}}
\begin{proof}
Under \ECDA, each proposer $i$ proposes to receiver $j$ in descending order of dating rates $\delta_{ij}$. When it is the turn of a pair $(i,j)$, if proposer $i$ still has residual cognitive capacity and has not been (even partially) rejected by receiver $j$, then proposer $i$ proposes to $j$ an amount
\begin{equation}
    p_{ij}=\min\{1, r_i\},
\end{equation}
where $r_i$ denotes $i$'s residual proposer capacity at that moment.

We prove by induction that after processing any prefix of pairs, the tentative matching produced by \ECDA coincides with the greedy allocation restricted to that prefix.
The base case is immediate since both start from the zero matrix.

For the induction step, let $(i,j)$ be the next pair in descending order of $\delta_{ij}$.
By the induction hypothesis, the current tentative matching equals the greedy allocation on all earlier pairs, so the residual capacities of proposer $i$, denoted by $r_i$, and receiver $j$, denoted by $b_j$, are
\begin{align}
    r_i &= c_i - \sum_{k\in[J]} M_{ik},\\
    b_j &= q_j - \sum_{l\in[I]} w_{l j} M_{l j},
\end{align}
where $M$ denotes the current tentative matching.

Because we process proposals in the global order $\succ$, any proposer who arrives after $(i,j)$ is weakly less preferred by $j$ than $i$.
Therefore, accepting the proposal from $i$ cannot trigger any future replacement of a previously accepted (more preferred) proposer by a less preferred one.
Hence, at this step receiver $j$ accepts proposer $i$ up to the remaining exposure budget $b_j$.
Formally, the accepted amount for $(i,j)$ is
\begin{equation}
    M_{ij} = \min\left\{1,\ r_i,\ \frac{b_j}{w_{ij}}\right\},
\end{equation}
This is exactly the greedy update rule in the theorem statement.
After setting $M_{ij}$ and updating $r_i \leftarrow r_i-M_{ij}$ and $b_j \leftarrow b_j-w_{ij}M_{ij}$, the tentative matching again coincides with the greedy allocation on the enlarged prefix.

By induction, after all pairs are processed, the final output of \ECDA coincides with the greedy algorithm, completing the proof.
\end{proof}

\end{document}